\documentclass{article}
\usepackage{amsmath,amssymb,amsthm,times,natbib,mathtools,url,geometry,hyperref}

\newtheorem{theorem}{Theorem}
\newtheorem{lemma}{Lemma}
\newtheorem{proposition}{Proposition}

\newcommand{\ri}{\ensuremath{\mathrm{i}}}
\newcommand{\E}{\ensuremath{E}}
\newcommand{\var}{\ensuremath{\operatorname{var}}}
\usepackage[font=small,width=0.88\textwidth]{caption}
\newcommand{\R}{\mathbb{R}}

\newcommand{\cov}{\ensuremath{\operatorname{cov}}}
\newcommand{\f}{\overline{f}}
\def\H{\mathrm{\scriptscriptstyle H}}
\def\T{\mathrm{\scriptscriptstyle T}}

\newcommand{\N}{\mathbb{N}}
\newcommand{\sn}{^{(n)}}
\newcommand{\btheta}{\theta}

\begin{document}
\title{The De-Biased Whittle Likelihood}
\author{Adam M. Sykulski$^1$, Sofia C. Olhede$^2$, Jonathan M. Lilly$^3$,\\Arthur P. Guillaumin$^2$ and Jeffrey J. Early$^3$}
\date{}
\maketitle
\begin{abstract}
The Whittle likelihood is a widely used and computationally efficient pseudo-likelihood. However, it is known to produce biased parameter estimates for large classes of models. We propose a method for de-biasing Whittle estimates for second-order stationary stochastic processes. The de-biased Whittle likelihood can be computed in the same $\mathcal{O}(n\log n)$ operations as the standard approach. We demonstrate the superior performance of the method in simulation studies and in application to a large-scale oceanographic dataset, where in both cases the de-biased approach reduces bias by up to two orders of magnitude, achieving estimates that are close to exact maximum likelihood, at a fraction of the computational cost. We prove that the method yields estimates that are consistent at an optimal convergence rate of $n^{-1/2}$, under weaker assumptions than standard theory, where we do not require that the power spectral density is continuous in frequency. We describe how the method can be easily combined with standard methods of bias reduction, such as tapering and differencing, to further reduce bias in parameter estimates. 
\let\thefootnote\relax\footnote{$^1$ Data Science Institute / Department of Mathematics and Statistics, Lancaster University, UK (email: a.sykuslki@lancaster.ac.uk)}
\let\thefootnote\relax\footnote{$^2$ Department of Statistical Science, University College London, UK}
\let\thefootnote\relax\footnote{$^3$ NorthWest Research Associates, Redmond WA, USA}
\end{abstract}
\section{Introduction}
This paper introduces an improved computationally-efficient method of estimating time series model parameters of second-order stationary processes. The standard approach is to maximize the exact time-domain likelihood, which in general has computational efficiency of order $n^2$ for regularly-spaced observations (where $n$ is the length of the observed time series) and produces estimates that are asymptotically efficient, converging at a rate of $n^{-1/2}$. A second approach is the method of moments, which in general has a computational efficiency of smaller order but with poorer statistical performance \citep[p.253]{brockwell2009time}, exhibiting both bias and often a higher variance. A third approach of approximating the exact likelihood, often referred to as {\em quasi}-, {\em pseudo-}, or {\em composite-likelihoods}, is receiving much recent attention across statistics, see e.g.~\cite{fan2014quasi} and \cite{guinness2017circulant}. In time series analysis, such likelihood approximations offer the possibility of considerable improvements in computational performance (usually scaling as order $n\log n$), with only small changes in statistical behaviour, see e.g.~\cite{dutta2015h} and \cite{anitescu2016inversion}. Here we introduce a pseudo-likelihood that is based on the Whittle likelihood \citep{whittle1953estimation} which we will show offers dramatic decreases in bias and mean-squared error in applications, yet with no significant increase in computational cost, and no loss in consistency or rate of convergence. We will refer to our pseudo-likelihood as the {\em de-biased} Whittle likelihood.

The Whittle likelihood of~\cite{whittle1953estimation} is a frequency-domain approximation to the exact likelihood. This method is considered a standard method in parametric spectral analysis on account of its order $n\log n$ computational efficiency \citep{choudhuri2004contiguity,fuentes2007approximate,matsuda2009fourier,krafty2013penalized,jesus2017inference}. However, it has been observed that the Whittle likelihood, despite its desirable asymptotic properties, may exhibit poor properties when applied to real-world, finite-length time series, particularly in terms of estimation bias~\citep{dahlhaus1988small,velasco2000whittle,contreras2006note}. Bias is caused by spectral blurring, sometimes referred to as {\em spectral leakage}~\citep{percival1993spectral}. Furthermore, when the time series model is specified in continuous time, but observed discretely, then there is the added problem of aliasing, which if unaccounted for will further increase bias in Whittle estimates. The challenge is to account for such sampling effects and {\em de-bias} Whittle estimates, while retaining the computational efficiency of the method. We here define such a procedure, which can be combined with tapering and appropriate differencing, as recommended by \cite{dahlhaus1988small} and \cite{velasco2000whittle}. This creates an automated procedure that incorporates all modifications simultaneously, without any hand-tuning or reliance on process-specific analytic derivations such as in~\cite{taniguchi1983second}.

We compare pseudo-likelihood approaches with simulated and real-world time series observations. In our example from oceanography, the de-biased Whittle likelihood results in parameter estimates that are significantly closer to maximum likelihood than standard Whittle estimates, while reducing the computational runtime of maximum likelihood by over a factor of 100, thus demonstrating the practical utility of our method. Additionally, the theoretical properties of our new estimator are studied under relatively weak assumptions, in contrast to~\cite{taniguchi1983second}, \cite{dahlhaus1988small}, and \cite{velasco2000whittle}. Taniguchi studies autoregressive processes that depend on a scalar unknown parameter. Dahlhaus studies processes whose spectral densities are the product of a known function with peaks that increase with sample size, and a latent spectral density that is twice continuously differentiable in frequency. Velasco and Robinson study processes that exhibit power-law behaviour at low frequencies and require continuous differentiability of the spectrum (at all frequencies except zero). Our assumptions on the spectral density of the time series will be milder. In particular, we will {\em not} require that the spectral density is continuous in frequency. Despite this, we are still able to prove consistency of de-biased Whittle estimates, together with a convergence rate matching the optimal $n^{-1/2}$.
\section{Definitions and Notation}
We shall assume that the stochastic process of interest is modelled in continuous time, however, the de-biased Whittle likelihood can be readily applied to discrete-time models, as we shall discuss later. We define $\{X_t\}$ as the infinite sequence obtained from sampling a zero-mean continuous-time real-valued process $X(t;\theta)$, where $\theta$ is a length-$p$ vector that specifies the process. That is, we let $X_t\equiv X(t\Delta;\theta)$, where $t$ is a positive or negative integer, $t=\ldots,-2,-1,0,1,2,\ldots$, and $\Delta>0$ is the sampling interval. If the process is second-order stationary, we define the autocovariance sequence by $s(\tau;\theta)\equiv\E\{X_tX_{t-\tau}\}$ for $\tau=\ldots,-2,-1,0,1,2,\ldots$, where $\E\{\cdot\}$ is the expectation operator. 
The power spectral density of  $\{X_t\}$ forms a Fourier pair with the autocovariance sequence, and is almost everywhere given by
\begin{equation}
f(\omega;\theta) = \Delta\sum_{\tau=-\infty}^\infty s(\tau;\theta)\exp(-\ri\omega\tau\Delta), \quad
s(\tau;\theta) = \frac{1}{2\pi}\int_{-\pi/\Delta}^{\pi/\Delta}f(\omega;\theta)\exp(\ri\omega\tau\Delta)d\omega.
\label{eq:covspec}
\end{equation}
As $\{X_t\}$ is a discrete sequence, its Fourier representation is only defined up to the Nyquist frequency $\pm\pi/\Delta$. Thus there may be departures between $f(\omega;\theta)$ and the continuous-time process spectral density, denoted as $\tilde{f}(\omega;\theta)$, which for almost all $\omega\in\mathbb{R}$ is given by
\begin{equation}
\tilde{f}(\omega;\theta)=\int_{-\infty}^\infty \tilde s(\lambda;\theta)\exp(-\ri\omega\lambda)d\lambda, \quad
\tilde s(\lambda;\theta) = \frac{1}{2\pi}\int_{-\infty}^\infty \tilde{f}(\omega;\theta)\exp(\ri\omega\lambda)d\omega.
\label{eq:covspec2}
\end{equation}
Here $\tilde s(\lambda;\theta)\equiv\E\{X(t)X(t-\lambda)\}$ (for $\lambda\in\mathbb{R}$) is the continuous-time process autocovariance, which is related to $s(\tau;\theta)$ via $\tilde s(\tau\Delta;\theta)=s(\tau;\theta)$, when $\tau$ is an integer. It follows that
\begin{equation}
\label{eq:alias}
f(\omega;\theta)= \sum_{k=-\infty}^{\infty } \tilde{f}\left(\omega + k \frac{2\pi }{\Delta} ;\theta\right), \quad \omega\in[-{\pi}/{\Delta},{\pi}/{\Delta}].
\end{equation}
Thus contributions to $\tilde{f}(\omega;\theta)$ outside of the range of frequencies $\pm\pi/\Delta$ are said to be {\em folded} or {\em wrapped} into $f(\omega;\theta)$. We have defined both $f(\omega;\theta)$ and $\tilde{f}(\omega;\theta)$, as both quantities are important in separating aliasing from other artefacts in spectral estimation. 

In addition to these theoretical quantities, we will also require certain quantities that are computed directly from a single length-$n$ sample $\{X_t\}_{t=1}^n$. A widely used, but statistically inconsistent, estimate of $f(\omega;\theta)$ is the periodogram, denoted $I(\omega)$, which is the squared absolute value of the Discrete Fourier Transform (DFT) defined as
\begin{equation}
I(\omega)\equiv\left|J(\omega)\right|^2,\quad J(\omega)\equiv\left(\frac{\Delta}{n}\right)^{1/2}\sum_{t=1}^nX_t\exp(-\ri\omega t\Delta),\label{eq:DFT}\quad \omega\in[-\pi/\Delta,\pi/\Delta].
\end{equation}
Note that $I(\omega)$ and $J(\omega)$ are taken to be properties of the observed realisation and are not regarded as functions of $\theta$.
\section{Maximum Likelihood and the Whittle Likelihood}\label{S:Whittle}
Consider the discrete sample $X=\left\{X\right\}_{t=1}^N$, which is organized as a length $n$ column vector. Under the assumption that $X$ is drawn from $X(t;\theta)$, the expected $n\times n$ autocovariance matrix is $C(\theta)\equiv\E\left\{XX^T\right\}$, where the superscript ``$T$'' denotes the transpose, and the components of $C(\theta)$ are given by $C_{ij}(\theta)=s\left(i-j;\theta\right)$. Exact maximum likelihood inference can be performed for Gaussian data by evaluating the log-likelihood \cite[p.254]{brockwell2009time} given by
\begin{equation}
\label{log-likelihood1}
\ell(\theta)\equiv-\log|C(\theta)|-X^T\,C^{-1}(\theta)\,X,
\end{equation}
where the superscript ``$-1$" denotes the matrix inverse, and $|C(\theta)|$ is the determinant of $C(\theta)$. We have removed additive and multiplicative constants not affected by $\theta$ in~\eqref{log-likelihood1}. The optimal choice of $\theta$ for our chosen model to characterize the sampled time series $X$ is then found by maximizing the likelihood function in~\eqref{log-likelihood1} leading to
\[
\hat{\theta} = \arg \max_{\theta\in \Theta} \ell(\theta),
\]
where $\Theta$ defines the parameter space of $\theta$. Because the time-domain maximum likelihood is known to have optimal properties, any other estimator will be compared with the properties of this quantity.

A standard technique for avoiding expensive matrix inversions is to approximate~\eqref{log-likelihood1} in the frequency domain, following the seminal work of \cite{whittle1953estimation}. This 
approach approximates $C(\theta)$ using a Fourier representation, and utilizes the special properties of Toeplitz matrices. Given the observed sampled time series $X$, the Whittle likelihood, denoted $\ell_W(\theta)$ is
\begin{equation}
\label{whittle_likelihood}
\ell_W(\theta)\equiv-\sum_{\omega\in \Omega}  \left\{\log \tilde f(\omega;\theta) +\frac{I(\omega)}{\tilde f(\omega;\theta)}\right\},
\end{equation}
where $\Omega$ is the set of discrete Fourier frequencies given by
\begin{equation}
\Omega \equiv (\omega_1,\omega_2,\ldots,\omega_n)=\frac{2\pi}{n\Delta}\left(-\lceil {n}/{2} \rceil +1,\ldots,-1,0,1,\ldots, \lfloor {n}/{2} \rfloor\right).
\label{fourier_frequencies}
\end{equation}
The subscript ``$W$" in $\ell_W(\theta)$ is used to denote ``Whittle." We have presented the Whittle likelihood in a discretized form here, as its usual integral representation must be approximated for finite-length time series. In general, if the summation in~\eqref{whittle_likelihood} is performed over subsets of $\Omega$ then the resulting procedure is {\em semi-parametric}.

The Whittle likelihood {\em approximates} the time-domain likelihood when all Fourier frequencies are used in~\eqref{fourier_frequencies}, i.e. $\ell(\theta)\approx \ell_W(\theta)$, and this statement can be made precise \citep{dzhaparidze1983spectrum}. Its computational efficiency is a significantly faster $\mathcal{O}(n\log n)$, versus $\mathcal{O}(n^2)$ for maximum likelihood, as the periodogram can be computed using the Fast Fourier Transform, thus explaining its popularity in practice.
\section{Modified pseudo-likelihoods}
The standard version of the Whittle likelihood~\eqref{whittle_likelihood} is calculated using the periodogram, $I(\omega)$. This spectral estimate, however, is known to be a biased measure of the continuous-time process's spectral density for finite samples, due to blurring and aliasing effects \citep{percival1993spectral}, as discussed in the introduction. Aliasing results from the discrete sampling of the continuous-time process to generate an infinite sequence, whereas blurring is associated with the truncation of this infinite sequence over a finite-time interval. The desirable properties of the Whittle likelihood rely on the {\em asymptotic} behaviour of the periodogram for large sample sizes. The bias of the periodogram for {\em finite} samples however, will translate into biased parameter estimates from the Whittle likelihood, as has been widely reported (see e.g.~\cite{dahlhaus1988small}).

In this section we propose a modified version of the Whittle likelihood in Section~\ref{Whittle-real} which de-biases Whittle estimates. Furthermore, tapering and differencing are two well-established methods for improving Whittle estimates \citep{dahlhaus1988small,velasco2000whittle}. In Sections~\ref{Whittle-taper} and \ref{SS:Differencing} we respectively outline how the de-biased Whittle likelihood can be easily combined with either of these procedures.

\subsection{The de-biased Whittle likelihood} \label{Whittle-real}
We introduce the following pseudo-likelihood function given by
\begin{eqnarray}
\label{discrete_fourier_likelihood}
\ell_D(\theta)&\equiv&-\sum_{\omega\in \Omega}  \left\{\log \overline{f}_n(\omega;\theta) +\frac{I(\omega)}{\overline{f}_n(\omega;\theta)}\right\},
\\
\label{FejerKernel}
\overline{f}_n(\omega;\theta)&\equiv&\int_{-\pi/\Delta}^{\pi/\Delta} f(\nu;\theta){\cal F}_{n,\Delta}\left(\omega-\nu\right)\;d\nu, \quad
{\cal F}_{n,\Delta}(\omega)\equiv\frac{\Delta}{2\pi n}\frac{\sin^2( n \omega\Delta/2)}{\sin^2(\omega\Delta/2)},
\end{eqnarray}
where the subscript ``$D$" stands for ``de-biased." Here $\tilde f(\omega;\theta)$ in~\eqref{whittle_likelihood} has been replaced by $\overline{f}_n(\omega;\theta)$, which is the {\em expected} periodogram, and may be shown to be given by the convolution of the true modelled spectrum with the Fej\'er kernel ${\cal F}_{n,\Delta}(\omega)$, such that $\overline{f}_n(\omega;\theta)\equiv\E\{I(\omega)\}$ \citep{bloomfield2004fourier}. We call~\eqref{discrete_fourier_likelihood} the \textit{de-biased} Whittle likelihood, where the set $\Omega$ is defined as in~\eqref{fourier_frequencies}.

Replacing the true spectrum $\tilde f(\omega;\theta)$ with the expected periodogram $\overline{f}_n(\omega;\theta)$ in~\eqref{discrete_fourier_likelihood} is a straightforward concept, however, our key innovation lies in formulating its efficient computation without losing $\mathcal{O}(n\log n)$ efficiency. If we directly use~\eqref{FejerKernel}, then this convolution would usually need to be approximated numerically, and could be computationally expensive. Instead we utilize the convolution theorem to express the frequency-domain convolution as a time-domain multiplication \citep[p.198]{percival1993spectral}, such that
\begin{eqnarray}
\overline{f}_n(\omega;\theta)=2\Delta\cdot\Re\left\{\sum_{\tau=0}^{n-1}\left(1-\frac{\tau}{n}\right)s(\tau;\theta)\exp(-\ri\omega\tau\Delta)\right\}-\Delta\cdot s(0;\theta),
\label{eq:meanperiodogram}
\end{eqnarray}
where $\Re\{\cdot\}$ denotes the real part. Therefore $\overline{f}_n(\omega;\theta)$ can be exactly computed at each Fourier frequency directly from $s(\tau;\theta)$, for $\tau=0,\ldots,n-1$, by using a Fast Fourier Transform in $\mathcal{O}(n \log n)$ operations. Care must be taken to subtract the variance term, $\Delta\cdot s(0;\theta)$, to avoid double counting contributions from $\tau=0$. Both aliasing and blurring effects are automatically accounted for in~\eqref{eq:meanperiodogram} in one operation; aliasing is accounted for by sampling the theoretical autocovariance function at discrete times, while the effect of blurring is accounted for by the truncation of the sequence to finite length, and the inclusion of the triangle function $\left(1-{\tau}/{n}\right)$ in the expression. 

The de-biased Whittle likelihood can also be used for discrete-time process models, as~\eqref{eq:meanperiodogram} can be computed from the theoretical autocovariance sequence of the discrete process in exactly the same way. Furthermore, the summation in~\eqref{discrete_fourier_likelihood} can be performed over a reduced range of frequencies to perform {\em semi-parametric} inference. If the analytic form of $s(\tau;\theta)$ is unknown or expensive to evaluate, then it can be approximated from the spectral density using Fast Fourier Transforms, thus maintaining $\mathcal{O} (n\log n)$ computational efficiency.

As an aside, we point out that computing the standard Whittle likelihood of~\eqref{whittle_likelihood} with the aliased spectrum $f(\omega;\theta)$ defined in~\eqref{eq:covspec}, without accounting for spectral blurring, would in general be more complicated than using the expected periodogram $\overline{f}_n(\omega;\theta)$. This is because the aliased spectrum $f(\omega;\theta)$ seldom has an analytic form for continuous processes, and must be instead approximated by either explicitly wrapping in contributions from $\tilde f(\omega;\theta)$ from frequencies higher than the Nyquist as in~\eqref{eq:alias}, or via an approximation to the Fourier transform in~\eqref{eq:covspec}. This is in contrast to the de-biased Whittle likelihood, where the effects of aliasing and blurring have been computed exactly in one single operation using~\eqref{eq:meanperiodogram}.  Thus addressing aliasing and blurring together using the de-biased Whittle likelihood is simpler and computationally faster to implement than accounting for aliasing alone.

\subsection{The de-biased tapered Whittle likelihood}\label{Whittle-taper}
To ameliorate spectral blurring of the periodogram, a standard approach is to pre-multiply the data sequence with a weighting function known as a data taper \citep{thomson1982spectrum}. The taper is chosen to have spectral properties such that {\em broadband} blurring will be minimized, and the variance of the spectral estimate at each frequency is reduced, although the trade-off is that tapering increases {\em narrowband} blurring as the correlation between neighbouring frequencies increases.

The {\em tapered} Whittle likelihood~\citep{dahlhaus1988small} corresponds to replacing the direct spectral estimator formed from $I(\omega)$ in~\eqref{eq:DFT} with one using the taper $h=\{h_t\}$
\begin{equation}
J(\omega;h)\equiv\Delta^{1/2}\sum_{t=1}^nh_tX_t\exp(-\ri\omega t\Delta), \quad I(\omega;h)\equiv\left|J(\omega;h)\right|^2,\quad \sum_{t=1}^nh_t^2=1,
\label{eq:sxh}
\end{equation}
where $h_t$ is real-valued. Setting $h_t=1/n^{1/2}$ for $t=1,\ldots n$ recovers the periodogram estimate of~\eqref{whittle_likelihood}. To estimate parameters we then maximize
\begin{equation}
\label{whittle_likelihood_taper}
\ell_T(\theta)\equiv-\sum_{\omega\in \Omega}  \left\{\log \tilde f(\omega;\theta) +\frac{I(\omega;h)}{\tilde f(\omega;\theta)}\right\},
\end{equation}
where the subscript ``$T$" denotes that a taper has been used. \cite{velasco2000whittle} demonstrated that for certain discrete processes it is beneficial to use this estimator, rather than the standard Whittle likelihood, for parameter estimation, particularly when the spectrum exhibits a high dynamic range. Nevertheless, tapering in itself will not remove all broadband blurring effects in the likelihood, because we are still comparing the tapered spectral estimate against the theoretical spectrum, and not against the expected tapered spectrum. Furthermore, there remain the issues of narrowband blurring, as well as aliasing effects with continuous sampled processes. 

A useful feature of our de-biasing procedure is that it can be naturally combined with tapering. To do this we define the likelihood given by
\begin{eqnarray}
\label{discrete_fourier_likelihood_taper}
\ell_{TD}(\theta)&\equiv&-\sum_{\omega\in \Omega}  \left\{\log \overline{f}_n(\omega;h,\theta)
+\frac{I(\omega;h)}{\overline{f}_n(\omega;h,\theta)}\right\},
\\ \nonumber
\overline{f}_n(\omega;h,\theta)&\equiv&\int_{-\pi/\Delta}^{\pi/\Delta} f(\nu;\theta){\cal H}_{\Delta}\left(\omega-\nu\right)\;d\nu,  \quad {\cal H}_{\Delta}(\omega)\equiv\Delta \left|\sum_{t=1}^n h_t \exp(-\ri\omega t \Delta) \right|^2,
\end{eqnarray}
with $I(\omega;h)$ as defined in~\eqref{eq:sxh} such that  $\overline{f}_n(\omega;h,\theta)\equiv\E\{I(\omega;h)\}$. We call $\ell_{TD}(\theta)$ the \textit{de-biased tapered Whittle likelihood} and $\overline{f}_n(\omega;h,\theta)$ the {\em expected tapered spectrum} which can be computed exactly and efficiently using a similar $\mathcal{O}(n\log n)$ calculation to~\eqref{eq:meanperiodogram} to find $\overline{f}_n(\omega;h,\theta)$ such that
\[
\overline{f}_n(\omega;h,\theta)=2\Delta\cdot\Re\left\{\sum_{\tau=0}^{n-1}s(\tau;\theta)\left(\sum_{t=1}^{n-\tau}h_th_{t+\tau}\right) \exp(-\ri\omega\tau\Delta)\right\}-\Delta\cdot s(0;\theta).
\]
Accounting for the particular taper used in $\overline{f}_n(\omega;h,\theta)$ accomplishes de-biasing of the tapered Whittle likelihood, just as using the expected periodogram does for the standard Whittle likelihood. The time-domain kernel $\sum_{t=1}^{n-\tau}h_th_{t+\tau}$ can be pre-computed  using FFTs or using a known analytical form. Then during optimization, an FFT of this fixed kernel multiplied by the autocovariance sequence is taken at each iteration. Thus the de-biased tapered Whittle likelihood is also an $\mathcal{O}(n\log n)$ pseudo-likelihood estimator.

Both the de-biased tapered and de-biased periodogram likelihoods have their merits, but the trade-offs are different with {\em nonparametric} spectral density estimation than they are with {\em parametric} model estimation. Specifically, although tapering {\em decreases} the variance of nonparametric estimates at each frequency, it conversely may {\em increase} the variance of estimated parameters. This is because the taper is reducing degrees of freedom in the data, which increases correlations between local frequencies. On the other hand, the periodogram creates broadband correlations between frequencies, especially for processes with a high dynamic range, which also contributes to variance in parameter estimates. We explore these trade-offs in greater detail in Section~\ref{S:Simulations}.

\subsection{The de-biased Whittle likelihood with differenced data}\label{SS:Differencing}
Another method of reducing the effects of blurring on Whittle estimates is to fit parameters to the differenced process. This was illustrated in \cite{velasco2000whittle}, where the Whittle likelihood was found to perform poorly with fractionally integrated processes that exhibited higher degrees of smoothness, but improved when working with the differenced process as this reduced the dynamic range of the spectrum and hence decreased broadband blurring. 

Whittle likelihood using the differenced process proceeds as follows. Define $Y(t;\theta) \equiv X(t+\Delta;\theta)-X(t;\theta)$ for the continuous-time differenced process, and $Y_t = X_{t+1}-X_{t}$ for the sampled process. The spectral density of $Y(t;\theta)$, denoted $\tilde f_{Y}(\omega;\theta)$, can be found from $\tilde f(\omega;\theta)$ via the relationship
\[
\tilde f_{Y}(\omega;\theta)=4\sin^2\left(\frac{\omega\Delta}{2}\right)\tilde f(\omega;\theta).
\]
Then the Whittle likelihood for differenced processes is performed by maximizing
\begin{equation}
\ell_W(\theta)\equiv-\sum_{\omega\in \Omega_Y}  \left\{\log \tilde f_{Y}(\omega;\theta) +\frac{I_{Y}(\omega)}{\tilde f_{Y}(\omega;\theta)}\right\},
\label{whittle_differenced}
\end{equation}
where $I_{Y}(\omega)$ is the periodogram of the sample $\{Y_t\}_{t=1}^{n-1}$. The set of Fourier frequencies $\Omega_Y$ are now ${2\pi} (-\lceil (n-1)/{2} \rceil +1,\ldots,-1,1,\ldots, \lfloor (n-1)/{2} \rfloor)/{(n-1)\Delta}$, where one degree of freedom has been lost by differencing, and a second has been lost as the zero frequency should be excluded because the spectral density is now equal to zero here. The de-biased Whittle likelihood is also straightforward to compute from $\{Y_t\}_{t=1}^{n-1}$ over the same set of Fourier frequencies $\Omega_Y$
\begin{align}
\label{debiased_differenced}
\ell_D(\theta)&\equiv-\sum_{\omega\in \Omega_Y}  \left\{\log \overline{f}_{n,Y}(\omega;\theta) +\frac{I_{Y}(\omega)}{\overline{f}_{n,Y}(\omega;\theta)}\right\},\\
\overline{f}_{n,Y}(\omega;\theta)&=2\Delta\cdot\Re\left\{\sum_{\tau=0}^{n-2}\left(1-\frac{\tau}{n}\right)s_{Y}(\tau;\theta)\exp(-\ri\omega\tau\Delta)\right\}-\Delta\cdot s_{Y}(0;\theta),\nonumber
\end{align}
where  $\overline{f}_{n,Y}(\omega;\theta)\equiv\E\{I_Y(\omega)\}$ and $s_{Y}(\tau;\theta)$ is the autocovariance of $Y_t$ where
\[
s_{Y}(\tau;\theta)=2s(\tau;\theta)-s(\tau+1;\theta)-s(\tau-1;\theta),
\]
from direct calculation. This likelihood remains an $\mathcal{O}(n\log n)$ operation to evaluate, as computing all required $s_{Y}(\tau;\theta)$ from $s(\tau;\theta)$ is $\mathcal{O}(n)$, and the rest of the calculation is as in~\eqref{eq:meanperiodogram}. Differencing and tapering can be easily combined in $\mathcal{O}(n\log n)$, with both the standard and de-biased Whittle likelihoods. Furthermore, differencing can be applied multiple times if required.

To see how differencing can reduce the variance of the estimators, we investigate the variance of the score of the de-biased Whittle likelihood, which (as derived in equation~\eqref{dynamic_range} of the Appendix material) can be bounded for each scalar $\theta_i$ by
\begin{equation}\label{e:scorebound}
\var\left\{\frac{1}{n}\frac{\partial }{\partial \theta_i}\ell_D(\theta)\right\}\ \leq \frac{f^2_{\max} \|\frac{\partial \overline{f}_n}{\partial \theta_i}\|^2_{\infty} }{n f_{\min}^4},
\end{equation}
where $\|{\partial \overline{f}_n}/{\partial \theta_i}\|^2_{\infty}$ is the upper bound on the partial derivative of the expected periodogram with respect to $\theta_i$, and $f_{\min}$ and $f_{\max}$ are upper and lower bounds on the spectral density respectively (assumed finite and non-zero). The significance of~\eqref{e:scorebound} is that the bound on the variance of the score is controlled by the dynamic range of the spectrum, as a high dynamic range will lead to large values of $f_{\max}^2 / f_{\min}^4$. This suggests that one should difference a process with steep spectral slopes, as this will typically reduce the dynamic range of the spectrum thus reducing $f_{\max}^2 / f_{\min}^4$, in turn decreasing the bound on the variance. Differencing multiple times however may eventually send $f_{\min}$ to zero, such that at some point the ratio will increase, in turn increasing the bound on the variance.
\section{Monte-Carlo Simulations}\label{S:Simulations}
All simulation results in Sections~\ref{S:Simulations} and \ref{S:Application} can be exactly reproduced in MATLAB, and all data can be downloaded, using the software available at \url{www.ucl.ac.uk/statistics/research/spg/software}. As part of the software we provide a simple package for estimating the parameters of any time series observation modelled as a second-order stationary stochastic process specified by its autocovariance.
\subsection{Comparing the standard and de-biased Whittle likelihood}
In this section we investigate the performance of the standard and de-biased Whittle likelihoods in a Monte Carlo study using observations from a Mat\'ern process~\citep{matern1960spatial}, as motivated by the simulation studies of \cite{anitescu2012matrix} who study the same process. The Mat\'ern process is a three-parameter continuous Gaussian process defined by its continuous-time unaliased spectrum
\begin{equation}\label{eq:matern}
\tilde f(\omega)=\frac{A^2}{(\omega^2+c^2)^\alpha}.
\end{equation}
The parameter $A$ sets the magnitude of the variability, $1/c>0$ is the damping timescale, and $\alpha>1/2$ controls the rate of spectral decay, or equivalently the smoothness or differentiability of the process. For large $\alpha$ the power spectrum exhibits a high dynamic range, and the periodogram will be a poor estimator of the spectral density due to blurring. Conversely, for small $\alpha$ there will be departures between the periodogram and the continuous-time spectral density because of aliasing. We will therefore investigate the performance of estimators over a range of $\alpha$ values. 
\begin{figure}
\centering{
\includegraphics[width=0.95\textwidth,trim={2cm 1cm 2cm 1.2cm},clip]{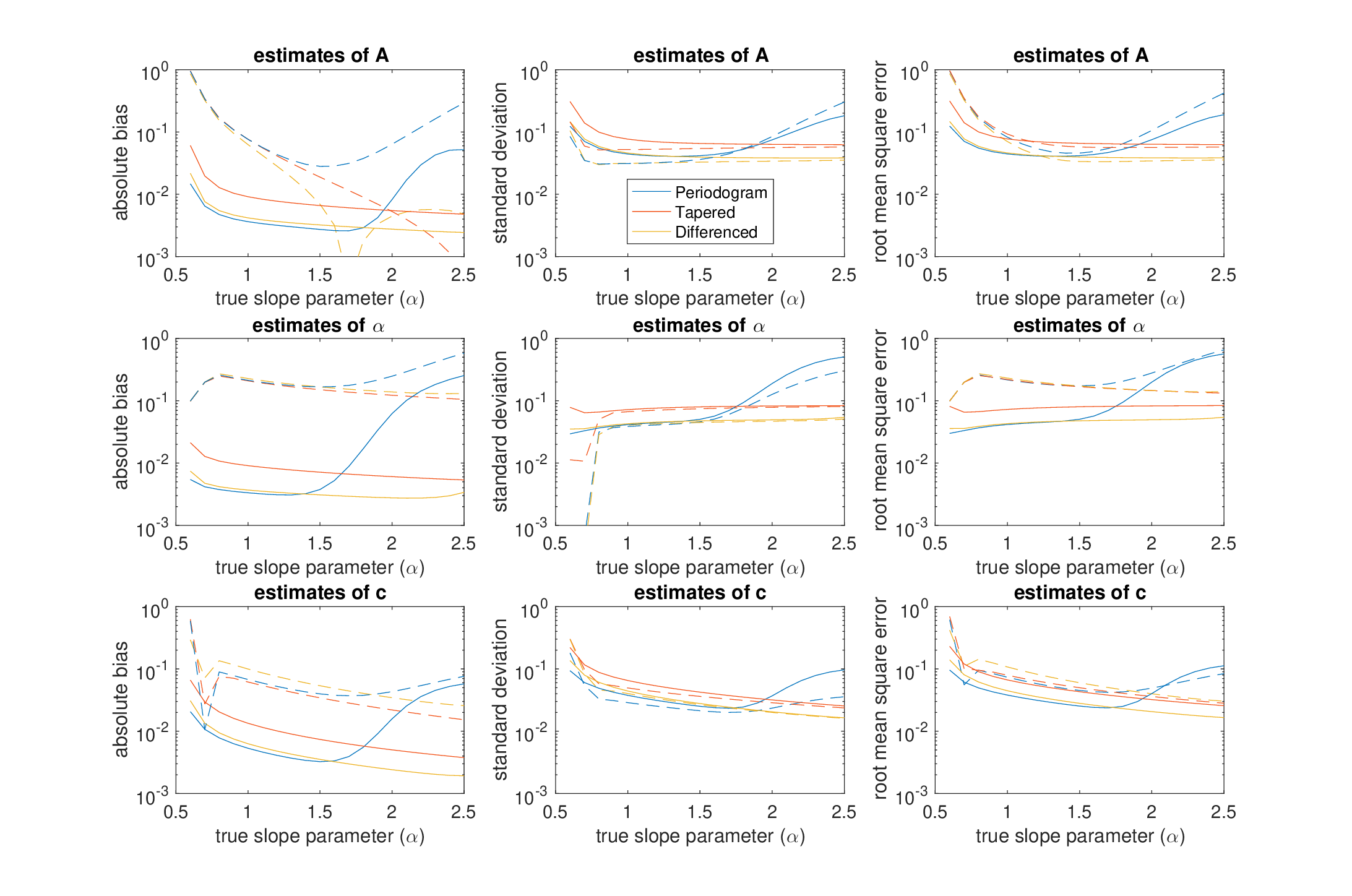}}
\caption{\label{Whittle2Fig}Absolute bias (left column), standard deviation (centre column) and root-mean-square-error (right column) of Mat\'ern parameter estimates using different forms of the standard and de-biased Whittle likelihoods. The colours correspond to different approaches: periodogram, tapered, or differenced (as indicated by the legend in the top centre panel). For each approach, the dashed lines are for standard Whittle approaches, and the solid lines are the corresponding de-biased approaches. The top row corresponds to estimates for the amplitude parameter $A$, the second row for the slope parameter $\alpha$, and the bottom row for the damping parameter $c$, as given in~\eqref{eq:matern}. In all panels the results are over a range of $\alpha$ values in increments of 0.1 from 0.6 to 2.5, translating to a spectral slope between $\omega^{-1.2}$ and $\omega^{-5}$.}
\end{figure}

In Fig.~\ref{Whittle2Fig} we display the bias and standard deviation of the different Whittle estimators for the three parameters $\{A,\alpha,c\}$ where $\alpha$ varies from $[0.6,2.5]$ in intervals of $0.1$. We fix $A=1$ and $c=0.2$, but estimate all three parameters assuming they are unknown. For each value of $\alpha$, we simulate 10,000 time series each of length $n=1000$, and use these replicated series to calculate biases and standard deviations for each estimator. We implement several different Whittle estimators: standard Whittle likelihood~\eqref{whittle_likelihood}, tapered Whittle likelihood~\eqref{whittle_likelihood_taper}, and differenced Whittle likelihood~\eqref{whittle_differenced}. In addition, for each of these we implement the de-biased version (equations \eqref{discrete_fourier_likelihood}, \eqref{discrete_fourier_likelihood_taper}, and \eqref{debiased_differenced}, respectively). The choice of data taper is the Discrete Prolate Spheroidal Sequence (DPSS) taper \citep{slepian1961prolate}, with bandwidth parameter equal to 4, where performance was found to be broadly similar across different choices of bandwidth (not shown). We also performed fits using combined differenced and tapered versions of both the standard and de-biased Whittle likelihoods, as discussed in Section~\ref{SS:Differencing}, and found that results were virtually identical to tapering without differencing (also not shown). The optimization is performed in MATLAB using \texttt{fminsearch}, and uses identical settings for all likelihoods. Initialized guesses for the slope and amplitude are found using a least squares fit in the range $[\pi/4\Delta,3\pi/4\Delta]$, and the initial guess for the damping parameter $c$ is set at a mid-range value of 100 times the Rayleigh frequency (i.e. $c=100\pi/n=\pi/10$.)

The first column in Fig.~\ref{Whittle2Fig} displays the absolute bias of each estimator for the three Mat\'ern parameters. The absolute biases are displayed on a log10 scale, and in many instances we can see bias reductions of over a factor of 10 in each of the parameters, representing over a 90\% bias reduction. The ``U'' shape over the range of $\alpha$ values, with the standard Whittle likelihood using the periodogram, corresponds to aliasing effects for small $\alpha$ and blurring effects for large $\alpha$. Differencing and tapering ameliorate the blurring effects for high $\alpha$, but not the aliasing effects for low $\alpha$. De-biased methods, particularly when combined with differencing, remove bias most consistently across the full range of $\alpha$ values.

The second column displays the standard deviations of the estimates. These are broadly comparable between the methods, although methods that use the periodogram suffer from reduced performance for high $\alpha$ due to broadband spectral blurring. The dip when estimating $\alpha$ for low values with standard methods is due to boundary effects in the optimization. Here the estimate of $\alpha$ cannot go below 0.5, and due to severe aliasing, the optimization typically converges to this lower bound when the true $\alpha$ is less than 0.7, such that the estimate is badly biased, but not variable. As we have used 10,000 replicates, the standard error of the reported biases can be observed by eye, by dividing the corresponding standard deviations by $10000^{1/2}=100$, meaning that the observed bias reductions using our approach appear highly significant.

The final column displays the root-mean-square-error (RMSE), thus combining information from the first two columns. With standard methods, the observed biases are in general larger than the standard deviations, so the shapes of the RMSE curves generally follow those of the biases. The de-biased methods are significantly less biased such that standard deviation is now the main contribution to RMSE. Overall, because bias tends to dominate variance with the standard
Whittle estimates, the de-biased methods improve upon the standard methods with only a few exceptions, and can reduce error by an order of magnitude.

\setlength{\tabcolsep}{6pt} 
\begin{table}
\caption{\label{aggregate}
Aggregated results from Fig.~\ref{Whittle2Fig}, averaging the percentage bias (relative to the true parameter values), standard deviation (SD), and root-mean-square-error (RMSE) across all estimates of $\{A,\alpha,c\}$, over the full range of $\alpha$ considered}
\centering
\begin{tabular}{lcr@{.}lr@{.}lr@{.}lc}
Inference Method & Eqn & \multicolumn{2}{c}{Bias} & \multicolumn{2}{c}{SD} & \multicolumn{2}{c}{RMSE}     \\
Standard Whittle (periodogram) & \eqref{whittle_likelihood} & 23&69\% & 10&34\%  & 26&66\% \\
De-Biased Whittle (periodogram) & \eqref{discrete_fourier_likelihood} & 3&96\% & 12&97\%  & 13&75\% \\
Standard Whittle (tapered) & \eqref{whittle_likelihood_taper} & 18&11\% & 12&23\%  & 23&12\% \\
De-Biased Whittle (tapered) & \eqref{discrete_fourier_likelihood_taper} & 2&60\% & 14&15\%  & 14&41\% \\
Standard Whittle (differenced) & \eqref{whittle_differenced} & 18&99\% & 9&33\%  & 22&09\% \\
De-Biased Whittle (differenced) & \eqref{debiased_differenced} & 1&19\% & 8&90\%  & 8&99\% \\
\end{tabular}
\end{table}

Finally, in Table~\ref{aggregate} we aggregate all information in Fig.~\ref{Whittle2Fig} to provide the average percentage bias (relative to the true parameter values), standard deviation, and RMSE for each likelihood estimator and each parameter combination. Of all the estimators, the de-biased Whittle likelihood using the differenced process performs best. Overall, of the three modifications to the standard Whittle likelihood---de-biasing, tapering and differencing---the de-biasing method proposed here is the single procedure that yields the greatest overall improvement in parameter estimation.

\subsection{Comparison with time-domain estimators}
Time-domain $\mathcal{O}(n\log n)$ pseudo-likelihood procedures have been proposed in \cite{anitescu2012matrix} (see also \cite{dutta2015h} and \cite{anitescu2016inversion}) who use Hutchinson trace estimators and circulant embedding techniques, removing the need to calculate a matrix inverse or determinant. We contrast the approach proposed here with that of \cite{anitescu2012matrix} using the MATLAB package ``ScalaGauss" supplied by those authors at \url{http://press3.mcs.anl.gov/scala-gauss/software/}. We use the same parameters of their example code for a Mat\'ern process, which estimates the damping parameter, and assumes the slope parameter is known, and the amplitude parameter is known up to a proportion of the damping parameter. The parameters used, when transformed into the form of~\eqref{eq:matern}, are $A=1.7725c,\alpha=1.5,c=0.0197$ and $n=1,024$. As the slope parameter is high, we fit the de-biased Whittle to the differenced process. We perform 10,000 repeats and report the results in Table~\ref{Table:MC3}. We include results for maximum likelihood, standard Whittle likelihood, and standard and de-biased tapered likelihoods.

Standard Whittle likelihood performs extremely poorly due to the blurring effects of using the periodogram. The de-biased Whittle likelihood and the method of \cite{anitescu2012matrix} return estimation errors that are very close to maximum likelihood. The method of \cite{anitescu2012matrix} however, requires an order of magnitude more processing time than the de-biased and standard Whittle likelihood. This is because the method involves many more steps in the procedure, such as generating random numbers to form the Hutchinson trace estimators. To speed up the \cite{anitescu2012matrix} method, we have included results with a modified version which uses only one Hutchinson trace estimator (as opposed to the 50 used in the example code). The method still remains slower than the de-biased Whittle likelihood, and now yields slightly worse estimation accuracy. The de-biased method appears to be the best method at combining the fit quality of time-domain maximum likelihood, with the speed of the standard Whittle method.

\setlength{\tabcolsep}{4pt} 
\begin{table}
\caption{\label{Table:MC3}
Percentage bias, standard deviation (SD), and root mean squared error (RMSE) of different methods when estimating the damping parameter, $\alpha$ of a Mat\'ern process. The experiment is repeated over 10,000 independently generated Mat\'ern series of length $n=1,024$ with parameters, $A=1.7725c,\alpha=1.5,c=0.0197$. CPU times are as performed on a 2.8 GHz Intel Core i7 processor}
\centering
\begin{tabular}{lcr@{.}lr@{.}lr@{.}lc}
Inference Method & Eqn & \multicolumn{2}{c}{Bias} & \multicolumn{2}{c}{SD} & \multicolumn{2}{c}{RMSE}  & CPU (sec.)    \\
Maximum likelihood & \eqref{log-likelihood1}  & 0&029\% & 2&204\%  & 2&204\%  &  4.257 \\
Standard Whittle (periodogram) & \eqref{whittle_likelihood} & 107&735\% & 101&357\%  & 147&916\% & 0.139 \\
De-Biased Whittle (differenced) & \eqref{debiased_differenced} & 0&030\% & 2&212\%  & 2&212\% & 0.157 \\
Standard Whittle (tapered) & \eqref{whittle_likelihood_taper} & 25&550\% & 20&459\%  & 32&731\% & 0.168 \\
De-Biased Whittle (tapered) & \eqref{discrete_fourier_likelihood_taper} & 0&023\% & 2&558\%  & 2&558\% & 0.198 \\
Anitescu et al. (normal version) &  & 0&029\% & 2&205\%  & 2&205\% & 1.998 \\
Anitescu et al. (faster version) &  & 0&035\% & 2&223\%  & 2&223\% & 0.438 \\
\end{tabular}
\end{table}
\section{Application to Large-Scale Oceanographic Data}\label{S:Application}
In this section we examine the performance of our method when applied to a real-world large-scale dataset, by analysing data obtained from the Global Drifter Program, which maintains a publicly-downloadable database of position measurements obtained from freely-drifting satellite-tracked instruments known as drifters (\texttt{http://www.aoml.noaa.gov/phod/dac/index.php}). In total over 23,000 drifters have been deployed, with interpolated six-hourly data available since 1979 and one-hourly data since 2005 (see \cite{elipot2016global}), with over 100 million data points available in total. The collection of such data is pivotal to the understanding of ocean circulation and its impact on the global climate system (see \cite{griffa2007lagrangian}); it is therefore essential to have computationally efficient methods for their analysis.

In Fig.~\ref{DrifterFig}, we display 50-day position trajectories and corresponding velocity time series for three drifters from the one-hourly data set, each from a different major ocean. These trajectories can be considered as complex-valued time series, with the real part corresponding to the east/west velocity component and the imaginary part corresponding to the north/south velocity component. We then plot the periodogram of the complex-valued series, which has different power at positive and negative frequencies, distinguishing directions of rotation on the complex plane \citep{schreier2010statistical}. The de-biased Whittle likelihood for complex-valued proper processes is exactly the same as~\eqref{discrete_fourier_likelihood}--\eqref{eq:meanperiodogram} (see also \cite{sykulski2016lagrangian}), where the autocovariance sequence of a complex-valued process $Z_t$ is $s(\tau;\theta)=\E\{Z_tZ^\ast_{t-\tau}\}$. For proper processes the complementary covariance is $r(\tau;\theta)=\E\{Z_tZ_{t-\tau}\}=0$ at all lags \citep{schreier2010statistical}, and can thus be ignored in the likelihood, as $s(\tau;\theta)$ captures all second-order structure in the zero-mean process.

We model the velocity time series as a complex-valued Mat\'ern process, with power spectral density given in~\eqref{eq:matern}, as motivated by~\cite{sykulski2016lagrangian} and~\cite{lilly2016fractional}. To account for a type of circular oscillations in each time series known as inertial oscillations, which create an off-zero spike on one side of the spectrum, we fit the Mat\'ern process semi-parametrically to the opposite ``non-inertial" side of the spectrum (as displayed by the red-solid line in the figure). We overlay the fit of the de-biased Whittle likelihood to the periodograms in Fig.~\ref{DrifterFig}. For a full parametric model of surface velocity time series, see~\cite{sykulski2016lagrangian}. We have selected drifters without noticeable tidal effects; for de-tiding procedures see \cite{pawlowicz2002classical}.

\begin{figure}
\centering{
\includegraphics[width=0.9\textwidth]{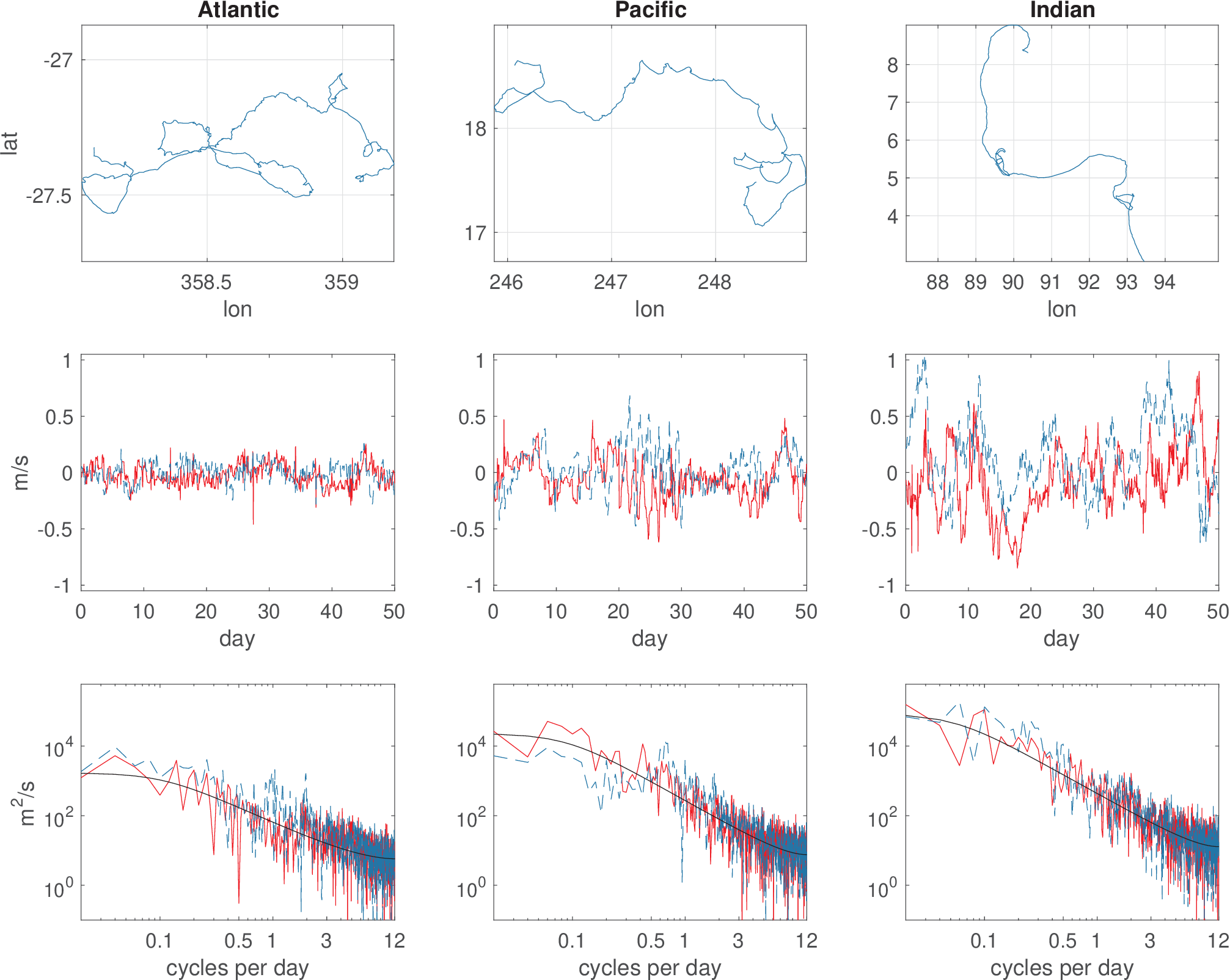}}
\caption{\label{DrifterFig}The top row displays 50-day trajectories of Drifter IDs: \#2339255 (Atlantic Ocean), \#49566 (Pacific Ocean), and \#43577 (Indian Ocean). The second row displays the east/west (red-solid) and north/south (blue-dashed) velocity time series for each trajectory. The third row displays the periodograms of the complex-valued velocity series, with the non-inertial side of the spectrum in red-solid, and the inertial side in blue-dashed. The expected periodogram, $\overline{f}_n(\omega;{\hat\theta})$, from the de-biased Whittle likelihood is overlaid in black.}
\end{figure}

\begin{table}
\caption{\label{DrifterTable}
Estimated Mat\'ern parameters using the maximum, de-biased Whittle, and standard Whittle, likelihoods for the velocity time series of Fig.~\ref{DrifterFig}. The parameters are given in terms of the damping timescale ($1/c$), the slope ($2\alpha$) and the diffusivity $(\kappa)$. CPU times are as performed on a 2.8 GHz Intel Core i7 processor}
\centering
\begin{tabular}{llr@{.}llr@{.}lr@{.}l}
Drifter & Inference & \multicolumn{2}{l}{Damping} & {Slope\quad} & \multicolumn{2}{l}{Diffusivity}  & \multicolumn{2}{l}{CPU (s)} \\
location & method & \multicolumn{2}{l}{(days)} & & \multicolumn{2}{l}{(m$^2$/s$\times 10^3$)}\quad & \multicolumn{2}{l}{ }  \vspace{2mm} \\
& Maximum likelihood\quad\quad &  10&65 & 1.460  & 0&49  &  20&89 \\
Atlantic & De-biased Whittle & 9&84 & 1.462  & 0&44 & 0&06 \\
& Standard Whittle & 30&19 & 1.097  & 0&65 & 0&02\vspace{2mm} \\
& Maximum likelihood & 10&63 & 1.829  & 5&09  &  31&65 \\
Pacific  & De-biased Whittle & 11&82 & 1.886  & 6&00 & 0&09 \\
& Standard Whittle & 19&59 & 1.566  & 7&18 & 0&02\vspace{2mm} \\
& Maximum likelihood & 21&76 & 1.825  & \quad30&47  &  30&37 \\
Indian   & De-biased Whittle & 19&91 & 1.802  & 22&71 & 0&11 \\
& Standard Whittle & 39&99 & 1.545  & 31&19 & 0&01 \\
\end{tabular}
\end{table}

We estimate the Mat\'ern parameters for each time series using the de-biased and regular Whittle likelihood, as well as exact maximum likelihood. The latter of these methods can be performed over only positive or negative frequencies by first decomposing the time series into analytic and anti-analytic components using the discrete Hilbert transform, see \cite{marple1999computing}, and then fitting the corresponding signal to an adjusted Mat\'ern autocovariance that accounts for the effects of the Hilbert transform. The details for this procedure are provided in the online code. 

The parameter estimates from the three likelihoods are displayed in Table~\ref{DrifterTable}, along with the corresponding CPU times. We reparametrize the Mat\'ern to output three important oceanographic quantities: the damping timescale, the decay rate of the spectral slope, and the diffusivity (which is the rate of particle dispersion) given by $\kappa \equiv A^2/4c^{2\alpha}$ \cite[eq.(43)]{lilly2016fractional}. From Table~\ref{DrifterTable} it can be seen that the de-biased Whittle and maximum likelihoods yield similar values for the slope and damping timescale, however, regular Whittle likelihood yields parameters that differ by around 15\% for the slope, and over 100\% for the damping timescale. These are consistent with the significant biases reported in our simulation studies in Section~\ref{S:Simulations}. The diffusivity estimates vary across all estimation procedures, and this variability is likely due to the fact that diffusivity is a measure of the spectrum at frequency zero, hence estimation is performed over relatively few frequencies. The time-domain maximum likelihood is over two orders of magnitude slower to execute than the de-biased Whittle likelihood. When this difference is scaled up to fitting all time series in the Global Drifter Program database, then time-domain maximum likelihood becomes impractical for such large datasets (taking years rather than days on the machine used in this example). The de-biased Whittle likelihood, on the other hand, retains the speed of Whittle likelihood, whilst returning estimates that are close to maximum likelihood. This section therefore serves as a proof of concept of how the de-biased Whittle likelihood is a useful tool for efficiently estimating parameters from large datasets.

\section{Properties of the De-Biased Whittle Likelihood}\label{S:Prop}
In this section, we establish consistency and optimal convergence rates for de-biased Whittle estimates with Gaussian processes. 
We will assume that the process is Gaussian in our proofs, however, formally we only require that the {\em Fourier transform} of the process is Gaussian. This is in general a weaker requirement. Processes that are non-Gaussian in the time domain may in fact have Fourier transforms with approximately Gaussian distributions for sufficiently large sample size. This is a consequence of a central limit theorem \cite[p.94]{brillinger2001time}, which also provides formal conditions when the Gaussian assumption is asymptotically valid. \cite{serroukh2000wavelet} provide practical examples which satisfy such conditions.

To show that de-biased Whittle estimates converge at on optimal rate, the main challenge is that although our pseudo-likelihood accounts for the bias of the periodogram, there is still present the broadband correlation between different frequencies caused by the leakage associated with the Fej\'er kernel. This is what prevents the de-biased Whittle likelihood from being exactly equal to the time-domain maximum likelihood for Gaussian data. To establish optimal convergence rates, we bound the asymptotic behaviour of this correlation. The statement is provided in Theorem~\ref{Thm1}, with the proof provided in the Appendix. The proof is composed of several lemmas which, for example, place useful bounds on the expected periodogram, the variance of linear combinations of the periodogram at different frequencies, and also the first and second derivatives of the de-biased Whittle likelihood. Together these establish that the de-biased Whittle likelihood is a consistent estimator with estimates that converge in probability at an optimal rate of $n^{-1/2}$, under relatively weak assumptions. 
\begin{theorem}\label{Thm1}
Assume that  $\{X_t\}$ is an infinite sequence obtained from sampling a zero-mean continuous-time real-valued process $X(t;\theta)$, which satisfies the following assumptions:
\begin{enumerate}
\item The parameter set $\Theta \subset \R^p$ is compact with a non-null interior, and the true length-$p$ parameter vector $\theta$ lies in the interior of $\Theta$.
\item Assume that for all $\theta\in\Theta$ and $\omega\in[-\pi,\pi]$, the spectral density of the sequence $\{X_t\}$ is bounded below by $f(\omega;\theta)\geq f_{\min}>0$, and bounded above by $f(\omega;\theta)\leq f_{\max}$.
\item If $\theta \neq \tilde\theta$, then there is a space of non-zero measure such that for all $\omega$ in this space $f(\omega;\theta)\neq f(\omega;\tilde\theta)$.
\item Assume that $f(\omega;\theta)$ is continuous in $\theta$ and Riemann integrable in $\omega$.
\item Assume that the expected periodogram $\overline{f}_n(\omega;\theta)$, as defined in~\eqref{FejerKernel}, has two continuous derivatives in $\theta$ which are bounded above in magnitude uniformly for all $n$, where the first derivative in $\theta$ also has $\Theta(n)$ frequencies in $\Omega$ that are non-zero.
\end{enumerate}
Then the estimator
\[
\hat{\theta}=\arg \max_{\theta\in \Theta}\ell_D(\theta),
\]
for a sample $\{X_t\}_{t=1}^n$, where $\ell_D(\theta)$ is the de-biased Whittle likelihood of~\eqref{discrete_fourier_likelihood}, satisfies
\[\hat{\theta}= \theta +{\cal O}_P\left( n^{-1/2}\right).\]
\end{theorem}
Standard theory shows that standard Whittle estimates are consistent with optimal convergence rates for Gaussian processes if the spectrum (and its first and second partial derivatives in $\theta$) are continuous in $\omega$ and bounded from above and below (see \cite{dzhaparidze1983spectrum}), as well as being twice continuously differentiable in $\theta$. In contrast, we have not required that the spectrum nor its derivatives are continuous in $\omega$; such that Theorem~\ref{Thm1} will hold for discontinuous spectra, as long as the other assumptions are satisfied such as Riemann integrability. As detailed in the Appendix, this is possible because the expectation of the score is now zero after de-biasing (equation~\eqref{zerobias} in the Appendix material, which would {\em not} be the case for the standard Whittle likelihood), such that we only need to consider the variance of the score and Hessian. To control these variances we make repeated use of a bound on the variance of linear combinations of the periodogram (Lemma~\ref{lemma=boundOnVariance})---a result previously established in~\cite[Theorem 3.1]{giraitis2013asymptotic} under a different set of assumptions.

It can be easily shown that the assumptions in Theorem~\ref{Thm1} are weaker than standard Whittle assumptions, despite requiring statements on the behaviour of the expected periodogram $\overline{f}_n(\omega;\theta)$ in Assumption 5. This is because if the spectral density $f(\omega;\theta)$ (and its first and second partial derivatives in $\theta$) are continuous in both $\omega$ and $\theta$, then it can be shown by applying the Leibniz' integration rule to the first and second derivatives of~\eqref{FejerKernel} in $\theta$, that $f(\omega;\theta)$ twice continuously differentiable in $\theta$ implies that $\overline{f}_n(\omega;\theta)$ is twice continuously differentiable in $\theta$. To show this we make use of~\cite[Prop 3.1]{stein2011fourier} which states that the convolution of two integrable and periodic functions is itself continuous. This result can also be used to show that $\overline{f}_n(\omega; \theta)$ is always continuous in $\omega$, even if $f(\omega; \theta)$ is not, as from~\eqref{FejerKernel} we see that $\overline{f}_n(\omega; \theta)$ is the convolution of $f(\omega; \theta)$ and the Fej\'er kernel---two functions which are integrable and $2\pi$-periodic in $\omega$. Therefore, not only does $\overline{f}_n(\omega; \theta)$ remove bias from blurring and aliasing, and is computationally efficient to compute, but it also has desirable theoretical properties leading to consistency and optimal convergence rates of de-biased Whittle estimates under weaker assumptions.

\section*{Appendix}
In the Appendix we prove that de-biased Whittle estimates converge at an optimal rate (Theorem~\ref{Thm1}). To prove Theorem~\ref{Thm1} we will first show that the debiased Whittle estimator is consistent (Proposition~\ref{theorem=consistency}) in a series of steps using eight Lemmas. Consistency will be established by showing that properties of the debiased Whittle estimator converge to that of the standard Whittle estimator. Then having established consistency, we establish the convergence rates via Lemma~\ref{prop:rate}, where the differences between the debiased and standard Whittle estimators will become especially clear. This allows us to establish optimal convergence rates under weaker assumptions in Theorem~\ref{Thm1}, where we shall not require that the spectral density is continuous in frequency.

Without loss of generality, we shall assume that the sampling interval is set to $\Delta=1$ in this section. We need to make the following assumptions on the stochastic process $X(t;\theta)$ to achieve consistency and optimal convergence rates:

\begin{enumerate}
\item The parameter set $\Theta \subset \R^p$ is compact with a non-null interior, and the true length-$p$ parameter vector $\theta$ lies in the interior of $\Theta$.
\item Assume that for all $\theta\in\Theta$ and $\omega\in[-\pi,\pi]$, the spectral density of the sequence $\{X_t\}$ is bounded below by $f(\omega;\theta)\geq f_{\min}>0$, and bounded above by $f(\omega;\theta)\leq f_{\max}$.
\item If $\theta \neq \tilde\theta$, then there is a space of non-zero measure such that for all $\omega$ in this space $f(\omega;\theta)\neq f(\omega;\tilde\theta)$.
\item Assume that $f(\omega;\theta)$ is continuous in $\theta$ and Riemann integrable in $\omega$.
\item Assume that the expected periodogram $\overline{f}_n(\omega;\theta)$, as defined in~\eqref{FejerKernel}, has two continuous derivatives in $\theta$ which are bounded above in magnitude uniformly for all $n$, where the first derivative in $\theta$ also has $\Theta(n)$ frequencies in $\Omega$ that are non-zero.
\end{enumerate}

We start with the following lemma which bounds the behaviour of the expected periodogram.
\begin{lemma}
	\label{lemma=boundexpectedperiodogram}
	For all $\theta\in\Theta$ and $n\in\N$, the expected periodogram $\overline{f}_n(\omega;\theta)$ is bounded below (by a positive real number), and above, independently of $n$ and $\theta$.
\end{lemma}
\begin{proof}
We start by noting that
\[
\overline{f}_n(\omega;\theta)=\int_{-\pi}^{\pi} 
f(\nu;\theta){\cal F}_{n}\left(\omega-\nu\right)\,d\nu, \quad{\cal F}_{n}\left(\omega\right)=\frac{1}{2\pi n}\frac{\sin^2\left(n\omega/2\right)}{\sin^2\left(\omega/2\right)},
\]
as given in~\eqref{FejerKernel} when $\Delta=1$. From Assumption 2 we have that $f(\nu;\theta)\geq f_{\min}>0$ and also that  $f(\nu;\theta)\leq f_{\max}\in\R$. It therefore follows that
\[
\overline{f}_n(\omega;\theta)\leq \int_{-\pi}^{\pi} 
{f}_{\max}{\cal F}_{n}\left(\omega-\nu\right)\,d\nu
={f}_{\max}\int_{-\pi}^{\pi} 
{\cal F}_{n}\left(\omega-\nu\right)\,d\nu={f}_{\max},
\]
such that the expected periodogram is upper bounded by $f_{\max}$. We also have that
\[
\overline{f}_n(\omega;\theta)\geq
\int_{-\pi}^{\pi} 
f_{\min}{\cal F}_{n}\left(\omega-\nu\right)\,d\nu={f}_{\min}\int_{-\pi}^{\pi} 
{\cal F}_{n}\left(\omega-\nu\right)\,d\nu={f}_{\min},
\]
such that the expected periodogram is lower bounded by $f_{\min}>0$.
 \end{proof}

Following the work of~\cite{taniguchi1979estimation} and \cite{guillaumin2017analysis}, we now introduce the following quantity
\begin{equation}
	D\sn\left(\gamma, g\right) = \frac{1}{n}\sum_{\omega\in\Omega}\left\{ \log \overline{f}_n(\omega;\gamma) + \frac{g(\omega)}{\overline{f}_n(\omega;\gamma)} \right\},
	\label{arthur1}
\end{equation}
for all $\theta\in\Theta$ and $n\in\N$.
We also define
\begin{equation}
	T\sn(g) = \arg\min_{\gamma\in\Theta}D\sn\left(\gamma, g\right).
	\label{arthur2}
\end{equation}
The minimum of $T\sn(g)$ for fixed $g$ is well defined since the set $\Theta$ is compact, and the function $D\sn\left(\gamma, g\right)$ is continuous in $\gamma$ (from Assumptions 1 and 5). However in cases where the minimum is not unique but exists at multiple parameter values, 
 $T\sn(g)$ will denote any of these values, chosen arbitrarily. We proceed with seven further lemmas that are required in proving Proposition~\ref{theorem=consistency} which establishes consistency, starting with Lemma~\ref{lemma=kappax} which we repeatedly use in the lemmas that follow.
\begin{lemma}\label{lemma=kappax}
The function $\kappa (x) = x-\log x$, defined on the set of positive real numbers, admits a global unique minimum for $x=1$ where it takes the value 1.
\end{lemma} 
\begin{proof}
This can be easily shown by taking the derivative of $\kappa(x)$.
 \end{proof}
 
\begin{lemma}
	\label{lemma=uniquenessOfMin}
	For all integer $n$ the quantity $T\sn(g)$ as defined in~\eqref{arthur1} and \eqref{arthur2}, satisfies $T\sn(\overline{f}_n(\omega;\theta)) = \theta$.
\end{lemma}
\begin{proof}
We have that for all $\gamma\in\Theta$
\begin{eqnarray*}
	D\left\{\gamma, \overline{f}_n(\omega;\theta)\right\} &=& \frac{1}{n}\sum_{\omega\in\Omega}
		\left\{ 
			\log \overline{f}_n(\omega;\gamma) + 
			\frac{\overline{f}_n(\omega;\theta)}{\overline{f}_n(\omega;\gamma)} 
		\right\} 
	\\	&=&\frac{1}{n} \sum_{\omega\in\Omega}
		\left[
		\log \overline{f}_n(\omega;\theta) +
		 \overbrace{\frac{\overline{f}_n(\omega;\theta)}{\overline{f}_n(\omega;\gamma)} - 
		\log\left\{  \frac{\overline{f}_n(\omega;\theta)}{\overline{f}_n(\omega;\gamma)}\right\}
		}^{\geq 1}
		\right]\\
		&\geq& \frac{1}{n}\sum_{\omega\in\Omega}
		\left\{
		\log \overline{f}_n(\omega;\theta)
		 +1\right\},
	\end{eqnarray*}
	where from Lemma~\ref{lemma=kappax} we have an equality if and only if $\overline{f}_n(\omega;\theta) = \overline{f}_n(\omega;\gamma)$ for all $\omega\in\Omega$, which is clearly satisfied for $\gamma=\theta$.
 \end{proof}

This shows that for all $n$, the function $\gamma\rightarrow D\left\{\gamma, \overline{f}_n(\cdot;\theta)\right\}$ reaches a global minimum at the true parameter vector $\theta$, although we have not proven any uniqueness properties at this stage. Now because $\overline{f}_n(\cdot;\theta)$ is changing with $n$, we require the following five lemmas.

\begin{lemma}
\label{lemma:localbehaviourfn}
		Define $\omega'\in[-\pi,\pi]$ such that $f(\omega;\gamma)$ is continuous at $\omega'$. For all $\epsilon>0$, there exists $\delta>0$ and an integer $m$ such that for all $n\geq m$, if $\omega''$ is such that $|\omega''-\omega'|\leq\delta$
		\begin{equation*}
			\left|
				\f_n(\omega'';\gamma) - f(\omega';\gamma)
			\right|
			\leq \epsilon.
		\end{equation*}
	\end{lemma}
	\begin{proof}
	By Assumption 4, $f(\omega;\gamma)$ is Riemann integrable in $\omega$, and therefore is continuous almost everywhere. It follows that there then exists an $\omega'\in[-\pi,\pi]$ such that $f(\omega;\gamma)$ is continuous at $\omega'$. Now let $\epsilon>0$.
	By continuity of $f(\omega;\gamma)$ at $\omega'$, there exists $\delta>0$ such that $|f(\omega'';\gamma)-f(\omega';\gamma)|\leq \epsilon/2$ for all $\omega''$ such that $|\omega''-\omega'|\leq\delta$.
		According to~\citep[p.71]{brockwell2009time} there exists an integer $m$ such that, for all $n\geq m$
		\begin{equation*}
			\int_{\frac{\delta}{2}\leq|\lambda|\leq\pi}{|F_n(\lambda)|}d\lambda\leq \frac{\epsilon}{2\|f\|_\infty}.
		\end{equation*}
		We then have for $\omega''$ such that $|\omega''-\omega'|\leq \delta/2$
		\begin{align*}
			\left|\f_n(\omega'';\gamma) - f(\omega';\gamma)\right| &= \left|\int_{-\pi}^\pi{\left\{f(\omega''-\lambda;\gamma)-f(\omega';\gamma)\right\} F_n(\lambda)}d\lambda\right|\\
			&\leq \int_{-\frac{\delta}{2}}^{\frac{\delta}{2}}{|f(\omega''-\lambda;\gamma)-f(\omega';\gamma)||F_n(\lambda)| }d\lambda\\
			&+  \int_{\frac{\delta}{2}\leq|\lambda|\leq\pi}{|f(\omega''-\lambda;\gamma)-f(\omega';\gamma)||F_n(\lambda)|}d\lambda.
		\end{align*}
		Observing that for $|\lambda|\leq\delta/2$ (in the first integral of the above equation) and given our choice of $\omega''$, we have by the triangle inequality $|(\omega''-\lambda)-\omega'|\leq |\omega''-\omega'|+|\lambda|\leq \delta$, such that $|f(\omega''-\lambda;\gamma)-f(\omega';\gamma)|\leq\epsilon/2$ and thus we obtain
		\[
			\left|\f_n(\omega'';\gamma) - f(\omega';\gamma)\right|\leq \frac{\epsilon}{2} \int_{-\frac{\delta}{2}}^{\frac{\delta}{2}}{|F_n(\lambda)| }d\lambda + 
			\|f\|_\infty\int_{\frac{\delta}{2}\leq|\lambda|\leq\pi}{|F_n(\lambda)|}d\lambda
			\leq \frac{\epsilon}{2}+\frac{\epsilon}{2} = \epsilon.
		\]
		This concludes the proof.
	 \end{proof}

\begin{lemma}
		\label{lemma:convergenceToIntegral}
		Recalling the definition of $D\sn\left\{\gamma, \f_n(\cdot;\btheta)\right\}$ in~\eqref{arthur1}, we have that
		\begin{equation}
			\label{eq:integral13060611}
			D\sn\left\{\gamma, \f_n(\cdot;\btheta)\right\} \rightarrow 
			\frac{1}{2\pi} \int_{-\pi}^{\pi}\left\{\log f(\omega;\gamma) + \frac{f(\omega;\btheta)}{f(\omega;\gamma)}\right\}d\omega.
		\end{equation}
	\end{lemma}
	\begin{proof}
		We have
		\[
		D\sn\left\{\gamma, \f_n(\cdot;\theta)\right\}  = \frac{1}{n}\sum_{\omega\in\Omega_n}\left\{\log \f_n(\omega;\gamma) + \frac{\f_n(\omega;\theta)}{\f_n(\omega;\gamma)}\right\}.
		\]
		Define 
		\[
			g_n(\omega) = \log \f_n\left\{k_n(\omega);\gamma\right\} + \frac{\f_n\left\{k_n(\omega);\theta\right\}}{\f_n\left\{k_n(\omega);\gamma\right\}},
		\]
		where $k_n(\omega) = \frac{\left\lfloor n \omega \right\rfloor}{n}$, i.e. $k_n(\omega)$ corresponds to the closest smaller Fourier frequency to $\omega$.
		Then,
		\[
			D\sn\left\{\gamma, \f_n(\omega;\theta)\right\} = \int_{-\pi}^\pi{g_n(\omega)d\omega}.
		\]
		We shall now use the bounded convergence theorem, for which we need to show that $g_n(\omega)$ converges almost everywhere. We recall $f(\omega;\theta)$ is continuous almost everywhere. Now take $\omega'\in[-\pi,\pi]$ such that $f(\omega;\theta)$ is continuous at $\omega'$. Let $\epsilon>0$. 
		Using Lemma~\ref{lemma:localbehaviourfn}, there exists $\delta$ and an integer $m$ such that, for all $\omega''$ satisfying $|\omega''-\omega'|\leq \delta$, and for all $n\geq m$, $|\f_n(\omega'')-f(\omega')|\leq \epsilon$.
		
		Additionally, $k_n(\omega')$ converges to $\omega'$ when $n$ goes to infinity, such that there exists an integer $p$ such that $|k_n(\omega')-\omega'|\leq\delta$ for $n\geq p$. Therefore, eventually in $n$, we have $|\f_n\left\{k_n(\omega');\theta\right\}-f(\omega';\theta)|\leq \epsilon$. Thus, $\f_n\left\{k_n(\omega');\theta\right\}$ converges to $f(\omega';\theta)$. As we have shown this for almost every $\omega'\in[-\pi,\pi]$, we have proved the point-wise convergence of $\f_n\left\{k_n(\omega);\theta\right\}$ to $f(\omega;\theta)$ almost everywhere with respect to $\omega$ on $[-\pi,\pi]$. The same reasoning shows the point-wise convergence of $\f_n\left\{k_n(\omega);\gamma\right\}$ to $f(\omega;\gamma)$ and that of $\log \f_n\left\{k_n(\omega);\gamma\right\}$ to $\log f(\omega;\gamma)$ almost everywhere with respect to $\omega$ on $[-\pi,\pi]$, as $f(\omega;\gamma)$ and $\f_n(\omega;\gamma)$ are bounded below. As the finite intersection of Lebesgue sets each having measure $2\pi$ is a Lebesgue set with measure $2\pi$, $g_n(\omega)$ converges point-wise almost everywhere to the integrand of the right-hand-side of~\eqref{eq:integral13060611}. Moreover, $g_n(\omega)$ is clearly upper bounded in absolute value by an integrable function according to Lemma~\ref{lemma=boundexpectedperiodogram}, such that we can apply the dominated convergence theorem and conclude that the sum on the left-hand-side of~\eqref{eq:integral13060611} converges to the integral on the right-hand-side of~\eqref{eq:integral13060611} as $n$ goes to infinity.
	 \end{proof}

\begin{lemma}
\label{lemma:gamma_nTotheta}
	If $\gamma\in\Theta$ and if $\{\gamma_n\}_{n\in\N}\in\Theta^{\N}$ is a sequence of parameter vectors converging to $\gamma$, then it follows that
	\[
			D\sn\left\{\gamma_n, \f_n(\omega;\theta)\right\} \rightarrow 
			\frac{1}{2\pi} \int_{-\pi}^{\pi}\left\{\log f(\omega;\gamma) + \frac{f(\omega;\theta)}{f(\omega;\gamma)}\right\}d\omega.
		\]
		\end{lemma}	
\begin{proof}
By the triangle inequality and having proved Lemma~\ref{lemma:convergenceToIntegral} we only need to prove that
\[
		D\sn\left\{\gamma_n, \f_n(\cdot;\theta)\right\}- D\sn\left\{\gamma, \f_n(\cdot;\theta)\right\},
		\]
		converges to zero.
This quantity can be written as
\[
	\frac{1}{n}\sum_{\omega\in\Omega_n}\left[\log \f_n(\omega;\gamma_n)- \log\f_n(\omega;\gamma) + \f_n(\omega;\theta)\left\{ \frac{1}{\f_n(\omega;\gamma_n)}-\frac{1}{\f_n(\omega;\gamma)} \right\}\right],
\]
which converges to zero because of the upper bound on the absolute derivative of $\partial \f_n (\omega;\gamma)/ \partial\gamma$ and lower bound for $\f_n(\omega;\gamma)$.
 \end{proof}

\begin{lemma}
\label{lemma=minimalValues}
If $\left\{\gamma_n\right\}_{n\in\N}\in\Theta^\N$ is a sequence of parameter vectors such that $D\left(\gamma_n, \overline{f}_n(\cdot;\theta)\right)-D\left(\theta, \overline{f}_n(\cdot;\theta)\right)$ converges to zero when $n$ goes to infinity, then $\gamma_n$ converges to $\theta$.
\end{lemma}
\begin{proof}
Let $(\gamma_n)_{n\in\N}$ be a sequence of parameter vectors such that
\begin{equation}
	\label{eq:convergenceToMinimum}
	D\sn\left\{\gamma_n, \overline{f}_n(\cdot;\btheta)\right\}-D\sn\left\{\btheta, \overline{f}_n(\cdot;\btheta)\right\} \rightarrow 0.
\end{equation}
We assume, with the intent to reach a contradiction, that the sequence $(\gamma_n)$ does not converge to $\btheta$. By compactness of $\Theta$, there exists an increasing function from the set of positive integers to the set of positive integers, denoted $\phi$, and $\gamma\in\Theta$ distinct from $\theta$, such that the sequence $\gamma_{\phi(n)}$ converges to $\gamma$ as $n$ goes to infinity.

Using Lemma~\ref{lemma:gamma_nTotheta}, we have that 
\[
D^{\phi(n)}\left\{\gamma_{\phi(n)}, \overline{f}_{\phi(n)}(\cdot;\theta)\right\} \rightarrow
\frac{1}{2\pi} \int_{-\pi}^{\pi}\left\{\log f(\omega;\gamma) + \frac{f(\omega;\btheta)}{f(\omega;\gamma)}\right\}d\omega.
\]
Similarly,
\[
D^{\phi(n)}\left\{\theta, \overline{f}_{\phi(n)}(\cdot;\theta)\right\} \rightarrow
\frac{1}{2\pi} \int_{-\pi}^{\pi}\left\{\log f(\omega;\theta) + 1\right\}d\omega.
\]

As $f(\omega;\gamma)$ and $f(\omega;\theta)$ are, by Assumption 3, distinct on a non-zero Lebesgue subset of $[-\pi, \pi]$, and using the properties of the function $x\rightarrow x - \log x$ from Lemma~\ref{lemma=kappax}, we have that the expression $(1/2\pi) \int_{-\pi}^{\pi}\left\{\log f(\omega;\gamma) + f(\omega;\btheta)/f(\omega;\gamma)\right\}d\omega$ is strictly larger than $(1/2\pi) \int_{-\pi}^{\pi}\left\{\log f(\omega;\theta) + 1\right\}d\omega$, and this is a contradiction with~\eqref{eq:convergenceToMinimum}. Therefore the sequence $\gamma_n$ must converge to $\theta$, and this concludes the proof.
 \end{proof}

We now show that the functions $D\left\{\gamma,\overline{f}_n(\omega;\theta)\right\}$ and $D\left\{\gamma, I(\omega)\right\}$, defined on $\Theta$, are asymptotically equivalent. For this, we first need the following lemma where we bound the asymptotic variance of linear combinations of the periodogram, a result previously established in \cite[Theorem 3.1]{giraitis2013asymptotic} and \cite{guillaumin2017analysis} under different sets of assumptions.

\begin{lemma}\label{lemma=boundOnVariance}
Assume that  $\{X_t\}$ is an infinite sequence obtained from sampling a zero-mean continuous-time real-valued process $X(t;\theta)$ with a spectral density 
$f(\omega;\theta)$ that is bounded above by the finite value $f_{\max}$. Additionally, assume that the deterministic function $a_n(\omega)$ has a magnitude that is bounded above by $a_{n,\max}$ for all $\omega\in\Omega$. Then linear combinations of values of the periodogram, $I(\omega)$, at different frequencies, have a variance that is upper bounded by
\[
\var\left\{\frac{1}{n}\sum_{\omega\in\Omega} a_n(\omega) I(\omega)\right\} \le \frac{a^2_{n,\max}f^2_{\max}}{n}.
\]
\end{lemma}
\begin{proof}
By definition, for $\omega,\omega'\in\Omega$
\begin{align*}
\cov\left\{ I(\omega),I(\omega')\right\} 
=\E \left\{I(\omega)I(\omega')\right\}
-\E \left\{I(\omega)\right\}\E\left\{I(\omega')\right\}.
\end{align*}
Then, using the fact that the Fourier transform of a Gaussian process is also Gaussian, we may use Isserlis' theorem \citep{isserlis1918formula} and so obtain that
\begin{align}
\nonumber
\E \left\{I(\omega)I(\omega')\right\}
&=\E\left\{ J(\omega)J^*(\omega)
 J(\omega')J^*(\omega')\right\}\\
\nonumber
&=\E\left\{ J(\omega)J^*(\omega)\right\}
\E\left\{ J(\omega')J^*(\omega')\right\}
+\E\left\{ J(\omega)J^*(\omega')\right\}
\E\left\{ J(\omega')J^*(\omega)\right\},
\end{align}
where $J(\omega)$ is the Discrete Fourier Transform. Thus it follows that
\begin{equation}
\cov\left\{I(\omega),I(\omega')\right\}=\E\left\{ J(\omega)J^*(\omega')\right\}
\E\left\{ J(\omega')J^*(\omega)\right\}
=\left|\cov\left\{ J(\omega),J(\omega')\right\}\right|^2\ge 0.
\label{supp1}
\end{equation}
Recalling that $C(\theta)=\E\left\{XX^\T\right\}$, and after defining the vector $G(\omega)=[\exp(\ri\omega t)]^\T$ for $t=1,\dots,n$, such that $J(\omega)=(1/n^{1/2})G^\H(\omega) X$, then we can represent the covariance of the Fourier transform as given by
\begin{equation}
\cov\left\{J(\omega),J(\omega')\right\}=\E\left\{ J(\omega)J^*(\omega')\right\}=\frac{1}{n}\E\left\{G^\H(\omega) XX^\T G(\omega')\right\}=\frac{1}{n}G^\H(\omega) C(\theta)G(\omega').
\label{supp2}
\end{equation}
Substituting~\eqref{supp2} into~\eqref{supp1}, it follows that
\begin{align}
\cov\left\{I(\omega),I(\omega')\right\}&=\frac{1}{n^2}G^\H(\omega) C(\theta)G(\omega')\left\{G^\H(\omega) C(\theta)G(\omega')\right\}^\H \nonumber \\
&= \frac{1}{n^2}G^\H(\omega) C(\theta)G(\omega')G^\H(\omega') C^\H(\theta) G(\omega),
\label{supp3}
\end{align}
and so~\eqref{supp3} is a positive quadratic form. 
We note that $a_n(\omega)a_n(\omega')\le |a_n(\omega)a_n(\omega')| \le a^2_{n,\max}$. Therefore these relationships, given~\eqref{supp1}, imply that
\begin{align}
\var\left\{\frac{1}{n}\sum_{\omega\in\Omega} a_n(\omega) I(\omega)\right\} &= \frac{1}{n^2}\sum_{\omega\in\Omega} \sum_{\omega'\in\Omega} a_n(\omega)a_n(\omega')\cov \left\{ I(\omega),I(\omega')\right\}\nonumber\\
&\le \frac{a^2_{n,\max}}{n^2}\sum_{\omega\in\Omega} \sum_{\omega'\in\Omega} \cov \left\{ I(\omega),I(\omega')\right\}
\nonumber\\
&\le \frac{a^2_{n,\max}}{n^4}\sum_{\omega\in\Omega} \sum_{\omega'\in\Omega} G^\H(\omega) C(\theta)G(\omega')G^\H(\omega') C^\H(\theta) G(\omega)
\nonumber\\
&= \frac{a^2_{n,\max}}{n^4}\sum_{\omega\in\Omega} G^\H(\omega) C(\theta)\left\{ \sum_{\omega'\in\Omega}G(\omega')G^\H(\omega')\right\} C^\H(\theta) G(\omega).
\label{supp4}
\end{align}
It can then easily be verified that
\[
\sum_{\omega'\in\Omega}G(\omega')G^\H(\omega') = n I_n,
\]
where $I_n$ is the $n\times n$ identity matrix and $\Omega$ is given in~\eqref{fourier_frequencies}. This follows as any off-diagonal term $\sum_{\omega\in\Omega} \exp\left\{\ri\omega(t-s)\right\}=0$ for $t\neq s$. Therefore~\eqref{supp4} simplifies to
\begin{equation}
\var\left\{\frac{1}{n}\sum_{\omega\in\Omega} a_n(\omega) I(\omega)\right\}\le \frac{a^2_{n,\max}}{n^3}\sum_{\omega\in\Omega} G^\H(\omega) C(\theta)C^\H(\theta) G(\omega).
\label{supp5}
\end{equation}
The matrix $C(\theta)$ is Hermitian, so it can be written $UDU^\H$ where $U$ is unitary and $D$ is the diagonal matrix consisting of the set of $n$ eigenvalues $\{\eta_k\}$, such that
\begin{align}
G^\H(\omega) C(\theta)C^\H(\theta) G(\omega)&=G^\H(\omega) UDU^\H UD^\H U^\H G(\omega)=G^\H(\omega) UDD^\H U^\H G(\omega) \nonumber \\ &\le  \eta_{\max}^2 G^\H(\omega) U U^\H G(\omega) = \eta_{\max}^2 \|G(\omega) \|_2^2=\eta_{\max}^2n,
\label{supp6}
\end{align}
where $\eta_{\max}$ denotes the maximal eigenvalue of $\{\eta_k\}$, and the last equality uses that $\|G(\omega)\|^2_2=n$. It is well known that for a Toeplitz covariance matrix $C(\theta)$ with eigenvalues $\{\eta_k\}$ and associated with the spectral density $f(\omega)$, then $\eta_{\max} <f_{\max}$, see for example \cite[p. 384]{lee1988digital}.
Therefore, by combining~\eqref{supp5} and \eqref{supp6} it follows that
\[
\var\left\{\frac{1}{n}\sum_{\omega\in\Omega} a_n(\omega) I(\omega)\right\}\le \frac{a^2_{n,\max}}{n^3}\sum_{\omega\in\Omega}\eta_{\max}n \le \frac{a^2_{n,\max}f^2_{\max}}{n},
\]
as $|\Omega|=n$, thus yielding the desired result.
 \end{proof}

Remembering that $\overline{f}_n(\omega;\theta) = \E\left\{I(\omega)\right\}$, we thus have 
\begin{equation}
\nonumber\frac{1}{n}\sum_{\omega\in\Omega}{a_n(\omega)I(\omega)} 
= \frac{1}{n}\sum_{\omega\in\Omega}{a_n(\omega)\overline{f}_n(\omega;\theta)} + \mathcal{O}_P\left(n^{-1/2}\right).
\end{equation}
We are now able to state a consistency theorem for our estimator $\hat{\theta}$.
\begin{proposition}\label{theorem=consistency}
Assume that  $\{X_t\}$ is an infinite sequence obtained from sampling a zero-mean continuous-time real-valued process $X(t;\theta)$ which satisfies Assumptions (1--5). Then the estimator
\[
\hat{\theta}=\arg \max_{\theta\in \Theta} \ell_D(\theta),
\]
for a sample $\{X_t\}_{t=1}^n$, where $\ell_D(\theta)$ is the de-biased Whittle likelihood of~\eqref{discrete_fourier_likelihood}, satisfies
\[
\hat{\theta} \overset{P}{\longrightarrow} \theta.
\]
\end{proposition}
\begin{proof}
	Denote $\overline{h}\sn(\gamma;\theta) = D\left\{\gamma, \overline{f}_n(\omega;\theta)\right\}$ and
	$\hat{h}\sn(\gamma) = D\left\{\gamma, I(\omega)\right\}$
	defined for any $\gamma\in\Theta$. We have
	\begin{eqnarray*}
		\overline{h}\sn(\gamma;\theta) - \hat{h}\sn(\gamma) &=&
		\frac{1}{n}\sum_{\omega\in\Omega}\left\{\log{\overline{f}_n(\omega;\gamma)}+\frac{\overline{f}_n(\omega;\theta)}{\overline{f}_n(\omega;\gamma)}-\log{\overline{f}_n(\omega;\gamma)}-\frac{I(\omega)}{\overline{f}_n(\omega;\gamma)}\right\}\\
		&=&\frac{1}{n}\sum_{\omega\in\Omega}{\frac{\overline{f}_n(\omega;\theta)-I(\omega)}{\overline{f}_n(\omega;\gamma)}}.
	\end{eqnarray*}
	We have shown in Lemma~\ref{lemma=boundexpectedperiodogram} that $\overline{f}_n(\omega;\gamma)$ is bounded below in both variables $\omega$ and $\gamma$ by a positive real number, independently of $n$. Therefore, making use of Lemmas~\ref{lemma=boundexpectedperiodogram} and \ref{lemma=boundOnVariance} we have
	\begin{equation}
	\label{eq=supgoestozero}
		\sup_{\gamma\in\Theta}\left|\overline{h}\sn(\gamma;\theta) - \hat{h}\sn(\gamma)\right| \stackrel{P}{\longrightarrow} 0, \ \ (n\rightarrow\infty),
	\end{equation}
	where $\stackrel{P}{\longrightarrow}$ indicates that the convergence is in probability, as the difference is of stochastic order $n^{-\frac{1}{2}}$.
	In particular~\eqref{eq=supgoestozero} implies that 
	\[
	\left|\min_\gamma \overline{h}\sn(\gamma;\theta) - \min_\gamma \hat{h}\sn(\gamma)\right| \leq \sup_{\gamma\in\Theta}\left|\overline{h}\sn(\gamma;\theta) - \hat{h}\sn(\gamma)\right| \stackrel{P}{\longrightarrow}0,
	\]
	i.e.
	\begin{equation}
	\label{eq=cvgzeroP1}
		\left|\overline{h}\sn\left[T\sn\left\{\overline{f}_n(\omega;\theta)\right\};\theta\right] -\hat{h}\sn\left[T\sn\left\{I(\omega)\right\}\right] \right| \stackrel{P}{\longrightarrow} 0.
	\end{equation}	
 Relation (\ref{eq=supgoestozero}) also implies that
	\begin{equation}
	\label{eq=cvgzeroP2}
	\left| \overline{h}\sn\left[T\sn\left\{I(\omega)\right\};\theta\right]-\hat{h}\sn\left[T\sn\left\{I(\omega)\right\}\right]  \right| \stackrel{P}{\longrightarrow} 0 ,
	\end{equation}
	such that using the triangle inequality,~\eqref{eq=cvgzeroP1} and~\eqref{eq=cvgzeroP2}, we get
	\begin{equation*}
		\left| \overline{h}\sn\left[T\sn\left\{I(\omega)\right\};\theta\right]  - \overline{h}\sn\left[T\sn\left\{\overline{f}_n(\omega;\theta)\right\};\theta\right] \right| \stackrel{P}{\longrightarrow} 0.
	\end{equation*}	
	We then obtain the stated proposition making use of Lemmas~\ref{lemma=uniquenessOfMin} and~\ref{lemma=minimalValues}.
 \end{proof}

Before proceeding to prove optimal convergence rates of de-biased Whittle estimates (Theorem~\ref{Thm1}), we require one further lemma.
\begin{lemma}\label{prop:rate}
Assume that  $\{X_t\}$ is an infinite sequence obtained from sampling a zero-mean continuous-time real-valued process $X(t;\theta)$ which satisfies Assumptions (1--5).  Then the pseudo-likelihood $\ell_D(\theta)$, defined in~\eqref{discrete_fourier_likelihood}, has first and second derivatives in each component of $\theta$, denoted $\theta_i$, for $i=1,\ldots,p$, that satisfy
\[
\frac{1}{n}\frac{\partial \ell_D(\theta)}{\partial \theta_i}={\cal O}_P\left(n^{-1/2}\right),
\]
and
\[
\frac{1}{n}\frac{\partial ^2\ell_D(\theta)}{\partial \theta_i^2}+\frac{1}{n}\sum_{\omega\in \Omega}\frac{\left\{\frac{\partial \overline{f}_n\left(\omega; \theta\right)}{
\partial \theta_i} \right\}^2}{\overline{f}_n^2\left(\omega; \theta\right)}= {\cal O}_P\left(n^{-1/2}\right),
\]
respectively, where
\[
\frac{1}{n}\sum_{\omega\in \Omega}\frac{\left\{\frac{\partial \overline{f}_n\left(\omega; \theta\right)}{\partial \theta_i} \right\}^2}{\overline{f}_n^2\left(\omega; \theta\right)}=\Theta(1).
\]
\end{lemma}

\begin{proof}
We start by evaluating the score directly from~\eqref{discrete_fourier_likelihood}
\begin{equation}
\frac{1}{n}\frac{\partial \ell_D(\theta)}{\partial \theta_i}=-\frac{1}{n}\sum_{\omega\in \Omega}\left\{\frac{\partial }{\partial \theta_i}\log\overline{f}_n\left(\omega; \theta\right)-\frac{\partial \overline{f}_n\left(\omega; \theta\right)}{\partial \theta_i} \frac{
I\left(\omega\right)}{\overline{f}_n^2\left(\omega; \theta\right)}
\right\}.
\label{eq:firstderiv}
\end{equation}
If we take the expectation of~\eqref{eq:firstderiv} then, recalling that $\overline{f}_n(\omega;\theta)=\E\left\{I(\omega)\right\}$, 
\begin{equation}\label{zerobias}
\E\left\{\frac{1}{n}\frac{\partial \ell_D(\theta)}{\partial \theta_i}\right\}=-\frac{1}{n}\sum_{\omega\in \Omega}\left\{\frac{\partial }{\partial \theta_i}\log\overline{f}_n\left(\omega; \theta\right)-
\frac{\frac{\partial }{\partial \theta_i}\overline{f}_n\left(\omega; \theta\right)}{\overline{f}_n\left(\omega; \theta\right)}
\right\}=0.
\end{equation}
Note that for finite sample sizes in general $\E\left\{I(\omega)\right\}\neq f(\omega;\theta)$, so the expectation of the score would not exactly be zero for the standard Whittle likelihood.

Furthermore, using Lemma~\ref{lemma=boundOnVariance} with $a_n(\omega)=\left\{\partial \overline{f}_n\left(\omega; \theta\right)/\partial \theta_i\right\}/\overline{f}_n^2(\omega;\theta)$, the variance of the score takes the form of
\begin{equation}
\var\left\{\frac{1}{n}\frac{\partial \ell_D(\theta)}{\partial \theta_i}\right\}=\var\left\{ \frac{1}{n}\sum_{\omega\in \Omega} \frac{\partial \overline{f}_n\left(\omega; \theta\right)}{\partial \theta_i} \frac{I(\omega)}{\overline{f}_n^2(\omega;\theta)}\right\}
\leq\frac{f^2_{\max} \|\frac{\partial \overline{f}_n}{\partial \theta_i}\|^2_{\infty} }{n f_{\min}^4},
\label{dynamic_range}
\end{equation}
where $\|\partial \overline{f}_n/\partial \theta_i\|_{\infty}=\sup_\omega \left\{\partial \overline{f}_n(\omega;\theta)/\partial \theta_i\right\}$ and $f_{\min}\le \overline{f}_n(\omega;\theta)$ from Lemma~\ref{lemma=boundexpectedperiodogram}. We can therefore conclude using Chebyshev's inequality \cite[pp.113--114]{papoulis1991probability} that we can fix $C>0$ such that
\[
\Pr \left\{ \frac{1}{n}\left|\frac{\partial \ell_D(\theta)}{\partial \theta_i}\right|\ge C\frac{f_{\max}\|\frac{\partial \overline{f}_n}{\partial \theta_i}\|_{\infty}  }{n^{1/2}f_{\min}^2} \right\}\le \frac{1}{C^2},
\]
where $\Pr\{\cdot\}$ denotes the probability of the argument. We may therefore deduce that
\[
\frac{1}{n}\frac{\partial \ell_D(\theta)}{\partial \theta_i}={\cal O}_P\left(n^{-1/2}\right).
\]
Next we examine the Hessian. Again from~\eqref{discrete_fourier_likelihood} we have
\begin{align*}
\frac{\partial^2 \ell_D(\theta)}{\partial \theta_i^2}&=-\sum_{\omega\in \Omega}\left\{\frac{\partial ^2}{\partial \theta_i^2}\log\overline{f}_n\left(\omega; \theta\right)+\frac{\partial ^2}{\partial \theta_i^2}\overline{f}_n^{-1}\left(\omega; \theta\right)
I\left(\omega\right)
\right\}\\
&=-\sum_{\omega\in \Omega}\left[ \frac{\frac{\partial^2 \overline{f}_n\left(\omega; \theta\right)}{
\partial \theta_i^2} }{\overline{f}_n\left(\omega; \theta\right)}-\frac{\left\{\frac{\partial \overline{f}_n\left(\omega; \theta\right)}{
\partial \theta_i} \right\}^2}{\overline{f}_n^2\left(\omega; \theta\right)}\right]-\sum_{\omega\in \Omega}\left[ 2\frac{\left\{\frac{\partial \overline{f}_n\left(\omega; \theta\right)}{
\partial \theta_i} \right\}^2}{\overline{f}_n^3\left(\omega; \theta\right)}-\frac{\frac{\partial^2 \overline{f}_n\left(\omega; \theta\right)}{
\partial \theta_i^2} }{\overline{f}_n^2\left(\omega; \theta\right)}\right]I\left(\omega\right),
\end{align*}
thus as $\E\left\{I(\omega)\right\}=\overline{f}_n(\omega;\theta)$ it follows that
\begin{equation}
\E \left\{\frac{1}{n}\frac{\partial^2 \ell_D(\theta)}{\partial \theta_i^2}\right\}=-\frac{1}{n}\sum_{\omega\in \Omega}\frac{\left(\frac{\partial \overline{f}_n\left(\omega; \theta\right)}{
\partial \theta_i} \right)^2}{\overline{f}_n^2\left(\omega; \theta\right)}.
\label{eq:riemann}
\end{equation}
As $\overline{f}_n(\omega;\theta)$ is bounded above and below (from Lemma~\ref{lemma=boundexpectedperiodogram}), and $\left\{\partial \overline{f}_n\left(\omega; \theta\right)/
\partial \theta_i \right\}^2$ is bounded above and also bounded below at $\Theta(n)$ frequencies (using Assumption 2)
we have
\begin{equation}
-\frac{1}{n}\sum_{\omega\in \Omega}\frac{\left\{\frac{\partial \overline{f}_n\left(\omega; \theta\right)}{
\partial \theta_i} \right\}^2}{\overline{f}_n^2\left(\omega; \theta\right)} = \Theta(1).
\label{eq:riemann2}
\end{equation}
Furthermore, we have that,
\[
\var\left\{\frac{1}{n}\frac{\partial ^2\ell_D(\theta)}{\partial \theta_i^2}\right\}=\var\left\{\frac{1}{n} \sum_{\omega\in \Omega}\left[ 2\frac{\left\{\frac{\partial \overline{f}_n\left(\omega; \theta\right)}{
\partial \theta_i} \right\}^2}{\overline{f}_n^3\left(\omega; \theta\right)}-\frac{\frac{\partial^2 \overline{f}_n\left(\omega; \theta\right)}{
\partial \theta_i^2} }{\overline{f}_n^2\left(\omega; \theta\right)}\right]I\left(\omega\right)\right\}.
\]
We define
\[
{\mathcal{S}_n}\left(\theta\right)=\sup_{\omega}\left| 2\frac{\left\{\frac{\partial \overline{f}_n\left(\omega; \theta\right)}{
\partial \theta_i} \right\}^2}{\overline{f}_n^3\left(\omega; \theta\right)}-\frac{\frac{\partial^2 \overline{f}_n\left(\omega; \theta\right)}{
\partial \theta_i^2} }{\overline{f}_n^2\left(\omega; \theta\right)}\right|.
\]
In this instance we can bound the variance, using Lemma~\ref{lemma=boundOnVariance}, by
\begin{equation}
\var\left\{\frac{1}{n}\frac{\partial ^2\ell_D(\theta)}{\partial \theta_i^2}\right\}\le \frac{{\mathcal{S}_n}^2\left(\theta\right)f^2_{\max}}{n}.
\label{eq:varbound}
\end{equation}
We can therefore conclude from \eqref{eq:riemann}--\eqref{eq:varbound}, and using Chebyshev's inequality, that
\[
\Pr \left\{ \frac{1}{n}\left|\frac{\partial ^2\ell_D(\theta)}{\partial \theta_i^2}+\frac{1}{n}\sum_{\omega\in \Omega}\frac{\left\{\frac{\partial \overline{f}_n\left(\omega; \theta\right)}{
\partial \theta_i} \right\}^2}{\overline{f}_n^2\left(\omega; \theta\right)} \right|\ge C \cdot \frac{{\mathcal{S}_n}\left(\theta\right)f_{\max}}{n^{1/2}} \right\}\le \frac{1}{C^2}.
\]
This yields the second result.
 \end{proof}

Lemma~\ref{prop:rate} shows the order of the first and second derivative of $\ell_D(\theta)$.
{\em Mutatis mutandis} we can show the corresponding results hold for $\nabla \ell_D(\theta)$
and $H(\theta)$, the Hessian matrix. We have now proved the ancillary results required to prove Theorem~\ref{Thm1}.
\begin{proof}
We let the $p$-vector ${\theta}'$ lie in a ball centred at the $p$-vector ${\theta}$ with radius $\|{\hat\theta}-\theta\|$ (this is a shrinking radius as Proposition~\ref{theorem=consistency} has shown consistency).
We additionally define the $p\times p$ Hessian matrix $H({\theta})$, having entries given by
\[
H_{ij}({\theta})=\frac{\partial ^2 \ell_D({\theta})}{\partial \theta_i \partial \theta_j}.
\]
Then as Proposition~\ref{theorem=consistency} has shown $\hat{\theta}\overset{P}{\rightarrow} \theta$ we
can write for some $\|{{\theta}}-{\theta}'\|\leq \|{{\theta}}-\hat{\theta}\|$,
applying the Taylor expansion of \cite[p.201]{brockwell2009time},
\begin{equation}
\frac{1}{n}\nabla\ell_D\left(\hat{{\theta}}\right)=\frac{1}{n}\nabla\ell_D\left({\theta}\right)+
\frac{1}{n}{{H}}\left({\theta}'\right)\left( \hat{{\theta}}-{{\theta}}\right).
\label{s:ball}
\end{equation}
We shall now understand the terms of this expression better. We note directly that 
\begin{equation}
\label{Hij}
H_{ij}({\theta})=-\sum_{\omega\in \Omega} \frac{\frac{\partial^2 \overline{f}_n\left(\omega; \theta\right)}{
\partial \theta_i \partial \theta_j} }{\overline{f}_n^2\left(\omega; \theta\right)}\left\{\overline{f}_n\left(\omega; \theta\right)-I\left(\omega\right)\right\}-\sum_{\omega\in \Omega}\frac{\frac{\partial \overline{f}_n\left(\omega; \theta\right)}{\partial\theta_i }\frac{\partial \overline{f}_n\left(\omega; \theta\right)}{\partial\theta_j}}{\overline{f}_n^3\left(\omega; \theta\right)}\left\{2I\left(\omega\right)-\overline{f}_n\left(\omega; \theta\right)\right\}.
\end{equation}
We see from this expression, coupled with Lemma~\ref{lemma=boundOnVariance} and
Chebyshev's inequality, that
\[
\frac{1}{n}\sum_{\omega\in \Omega} \frac{\frac{\partial^2 \overline{f}_n\left(\omega; \theta\right)}{
\partial \theta_i \partial \theta_j} }{\overline{f}_n^2\left(\omega; \theta\right)}I\left(\omega\right)
-\frac{1}{n}\sum_{\omega\in \Omega} \frac{\frac{\partial^2 \overline{f}_n\left(\omega; \theta\right)}{
\partial \theta_i \partial \theta_j} }{\overline{f}_n^2\left(\omega; \theta\right)}\overline{f}_n\left(\omega\right)
\overset{P}{\rightarrow}0,\]
such that the limiting behaviour of the second partial derivatives need not be determined, as they are by assumption finite (Assumption 5).
Then if we define \[
{\cal H}_n({{\theta}})\equiv\frac{1}{n}\E\left\{H(\theta)\right\},
\]
and by taking expectations of~\eqref{Hij} we see that
\[
{\cal H}_{n,ij}({{\theta}})= -\frac{1}{n}\sum_{\omega\in \Omega}\frac{
\frac{\partial \overline{f}_n\left(\omega; \theta\right)}{\partial\theta_i }
\frac{\partial \overline{f}_n\left(\omega; \theta\right)}{\partial\theta_j}}{\overline{f}_n^2\left(\omega; \theta\right)}=O(1),
\]
as we have already noted that both $\overline{f}_n^{-2}(\omega; \theta)$ and  $\partial \overline{f}_n(\omega;\theta)/\partial \theta_i$ are bounded and Riemann integrable as per Lemma~\ref{prop:rate}. 
Furthermore,
\[
{\cal H}_{n,ii}({{\theta}})= -\frac{1}{n}\sum_{\omega\in \Omega}\frac{
\left(\frac{\partial \overline{f}_n\left(\omega; \theta\right)}{\partial\theta_i }
\right)^2}{\overline{f}_n^2\left(\omega; \theta\right)}= \Theta(1),
\]
as $\overline{f}_n^2\left(\omega; \theta\right)$ is bounded above and below independently of $n$, and $\left(\partial \overline{f}_n\left(\omega; \theta\right)/\partial\theta_i 
\right)^2$ is bounded above, and bounded below for $\Theta(n)$ frequencies (from Assumption 5).

We also note that the elements of the matrix $H(\theta)$ take the form of linear combinations of $I(\omega)$ and so the element-wise extension of Lemma~\ref{lemma=boundOnVariance} applies. This means
\begin{equation}
\label{randomexp}
\frac{1}{n}H(\theta)={\cal H}_n({{\theta}})+{\cal O}_P(n^{-1/2}).
\end{equation}
From Proposition~\ref{theorem=consistency} we then observe that
\[
\|{{\theta}}-{\theta}'\|\leq \|{{\theta}}-\hat{\theta}\|=o_P(1).
\]
By applying the Taylor expansion of \cite[p.201]{brockwell2009time}, we observe that 
\[
	\overline{f}_n\left(\omega; \theta'\right) = \overline{f}_n\left(\omega; \theta\right) + o_P(1), \quad \frac{\partial \overline{f}_n\left(\omega; \theta'\right)}{\partial\theta_i } = \frac{\partial \overline{f}_n\left(\omega; \theta\right)}{\partial\theta_i }+ o_P(1),
\]
where neither of the $o_P(1)$ terms depend on $\omega$ or $n$ because of the upper bound on the magnitude of the first and second derivatives of $\overline{f}_n\left(\omega; \theta\right)$ with respect to $\theta$ in Assumption 5. 
Therefore, and because $\overline{f}_n\left(\omega; \theta\right)$ is bounded below from Assumption 2,
\[
\frac{1}{n}\sum_{\omega\in\Omega}\left\{\frac{\frac{\partial \overline{f}_n\left(\omega; \theta'\right)}{\partial\theta_i }\frac{\partial \overline{f}_n\left(\omega; \theta'\right)}{\partial\theta_j}}{\overline{f}^2_n\left(\omega; \theta'\right)}\right\} 
=
\frac{1}{n}\sum_{\omega\in\Omega}\left\{\frac{\frac{\partial \overline{f}_n\left(\omega; \theta\right)}{\partial\theta_i }\frac{\partial \overline{f}_n\left(\omega; \theta\right)}{\partial\theta_j}}{\overline{f}^2_n\left(\omega; \theta\right)}\right\} 
+o_P(1),
\]
 since $\theta'$ converges to $\theta$. Writing the Hessian at $\theta'$ as
\[
\frac{1}{n}H_{ij}({\theta'})=\frac{1}{n}\sum_{\omega\in\Omega}\left[\mathcal{M}_n(\omega;\theta')\left\{\overline{f}_n\left(\omega; \theta'\right)-I(\omega)\right\}
+ \frac{\frac{\partial \overline{f}_n\left(\omega; \theta'\right)}{\partial\theta_i }\frac{\partial \overline{f}_n\left(\omega; \theta'\right)}{\partial\theta_j}}{\overline{f}_n\left(\omega; \theta'\right)}\right],
\]
where 
\[
\mathcal{M}_n(\omega;\theta')=\frac{
\frac{\partial^2 \overline{f}_n\left(\omega; \theta'\right)}
{\partial \theta_i \partial \theta_j}
\overline{f}_n\left(\omega; \theta'\right)
-2\frac{\partial \overline{f}_n\left(\omega; \theta'\right)}{\partial\theta_i }
\frac{\partial \overline{f}_n\left(\omega; \theta'\right)}{\partial\theta_j}
}
{\overline{f}_n^3\left(\omega; \theta'\right)},
\]
is bounded according to our set of assumptions, we then observe using the triangle inequality that 
\begin{align*}
&
\left|
\frac{1}{n}\sum_{\omega\in\Omega}{\mathcal{M}_n(\omega;\theta')\left\{\overline{f}_n\left(\omega; \theta'\right)-I(\omega)\right\}}
\right|\\ &=
\left|\frac{1}{n}\sum_{\omega\in\Omega}{\mathcal{M}_n(\omega;\theta')\left\{\overline{f}_n\left(\omega; \theta'\right)-\overline{f}_n\left(\omega; \theta\right) + \overline{f}_n\left(\omega; \theta\right)- I(\omega)\right\}}\right|\\
&\leq 
\left|\frac{1}{n}\sum_{\omega\in\Omega}{\mathcal{M}_n(\omega;\theta')\left\{\overline{f}_n\left(\omega; \theta'\right)-\overline{f}_n\left(\omega; \theta\right)\right\}}\right| + 
\left|\frac{1}{n}\sum_{\omega\in\Omega}{\mathcal{M}_n(\omega;\theta')\left\{\overline{f}_n\left(\omega; \theta'\right)- I(\omega)\right\}}\right|.
\end{align*}
The first sum converges to zero in probability given the bound of the derivative of $\overline{f}_n\left(\omega; \theta\right)$ with respect to $\theta$ from Assumption 5, and the second sum converges to zero according to Lemma~\ref{lemma=boundOnVariance}.
It follows that
\begin{equation}
\frac{1}{n}{{H}}\left({\theta}'\right)-\frac{1}{n}{{H}}\left({\theta}\right)=o_P(1).
\label{s:Hessian}
\end{equation}

Starting from~\eqref{s:ball}, using Lemma~\ref{prop:rate}, and substituting in~\eqref{randomexp} and~\eqref{s:Hessian}
we obtain
\begin{align*}
 \hat{{\theta}}-{{\theta}}&=-\left\{\frac{1}{n}{{H}}\left({\theta}\right)+o_P(1)\right\}^{-1}\frac{1}{n}\nabla\ell_D\left({\theta}\right)\nonumber\\
&=-\left\{{\cal H}_n({{\theta}})+{\cal O}_P\left(n^{-1/2}\right)+o_P(1)\right\}^{-1}\left\{{\cal O}_P\left(n^{-1/2}\right)\right\}\nonumber\\
&={\cal H}_n({{\theta}})^{-1}\left\{{I}+o_P(1)+{\cal O}_P\left(n^{-1/2}\right)\right\} \left\{{\cal O}_P\left(n^{-1/2}\right)\right\}\nonumber\\
&={\cal O}_P\left(n^{-1/2}\right),
\end{align*} 
which yields the result we require.
 \end{proof}

\end{document}